%% file: QuantumLightning.tex
\documentclass[letterpaper,11pt]{article} 
\usepackage[letterpaper,hmargin=1in,vmargin=1in]{geometry} 
\usepackage{fullpage}
\usepackage[pdfstartview=FitH,colorlinks,linkcolor=blue,citecolor=blue]{hyperref} 

\usepackage[T1]{fontenc}
\usepackage{amsthm} 
\usepackage{color} 
\usepackage{times} 
\usepackage{microtype,pdfsync} 
\usepackage{amsmath,amsfonts,amssymb}
\usepackage{url} 
\usepackage{xspace,paralist}
\usepackage{lmodern}
\usepackage{float}
\floatstyle{boxed}
\restylefloat{figure}
\usepackage{multirow}
\usepackage{savesym}
\usepackage[weather]{ifsym}
\savesymbol{Sun}
\savesymbol{Lightning}
\usepackage{marvosym}
\usepackage{graphicx}

\usepackage{xfrac}

\input{macros}

\floatstyle{plain}
\restylefloat{figure}

\begin{document}

\title{Quantum Lightning Never Strikes the Same State Twice \\\vspace{0.3em} \large{Or: Quantum Money from Cryptographic Assumptions}}
\author{{\sc Mark Zhandry} \\
Princeton University\\
    	{\tt mzhandry@princeton.edu} }

\date{}
\maketitle

\begin{abstract} Public key quantum money can be seen as a version of the quantum no-cloning theorem that holds even when the quantum states can be verified by the adversary.  In this work, investigate \emph{quantum lightning}, a formalization of ``collision-free quantum money'' defined by Lutomirski et al. [ICS'10], where no-cloning holds \emph{even when the adversary herself generates the quantum state to be cloned}.  We then study quantum money and quantum lightning, showing the following results:

\begin{itemize}
	\item We demonstrate the usefulness of quantum lightning beyond quantum money by showing several potential applications, such as generating random strings with a proof of entropy, to completely decentralized cryptocurrency without a block-chain, where transactions is instant and local.
	
	\item We give win-win results for quantum money/lightning, showing that either signatures/hash functions/commitment schemes meet very strong recently proposed notions of security, or they yield quantum money or lightning.  Given the difficulty in constructing public key quantum money, this gives some indication that natural schemes do attain strong security guarantees.
	
	\item We construct quantum lightning under the assumed multi-collision resistance of random degree-2 systems of polynomials.  Our construction is inspired by our win-win result for hash functions, and yields the first plausible standard model instantiation of a \emph{non-collapsing} collision resistant hash function.  This improves on a result of Unruh [Eurocrypt'16] that requires a quantum oracle.  
	
	\item We show that instantiating the quantum money scheme of Aaronson and Christiano [STOC'12] with \emph{indistinguishability obfuscation} that is secure against quantum computers yields a secure quantum money scheme.  This construction can be seen as an instance of our win-win result for signatures, giving the first separation between two security notions for signatures from the literature.  
\end{itemize}

Thus, we provide the first constructions of public key quantum money from several cryptographic assumptions.  Along the way, we develop several new techniques including a new precise variant of the no-cloning theorem.
\end{abstract}

\clearpage
\setcounter{page}{1}

\section{Introduction}

\input{intro}

\section{Preliminaries}

\input{prelim}

\section{Quantum Lightning}

\input{qlightning}

\section{Win-win Results}

\input{winwin}

\section{Constructing Quantum Lightning}

\input{constr}

\section{Quantum Money from Obfuscation}

\input{qmoney}

\bibliographystyle{alpha}
\bibliography{abbrev0,crypto,bib}

\appendix

\end{document}

%% file: macros.tex
\theoremstyle{plain} 

\newtheorem{theorem}{Theorem}[section] 
\newtheorem{lemma}[theorem]{Lemma} 
 
\newtheorem{corollary}[theorem]{Corollary}
\newtheorem{claim}[theorem]{Claim}

\theoremstyle{definition} 

\newtheorem{definition}[theorem]{Definition}
\newtheorem{remark}[theorem]{Remark} 

\newtheorem{assumption}[theorem]{Assumption}

\newcommand{\thunderstorm}{\ensuremath{%
		\mathchoice{\raisebox{-2pt}{\includegraphics[height=2ex]{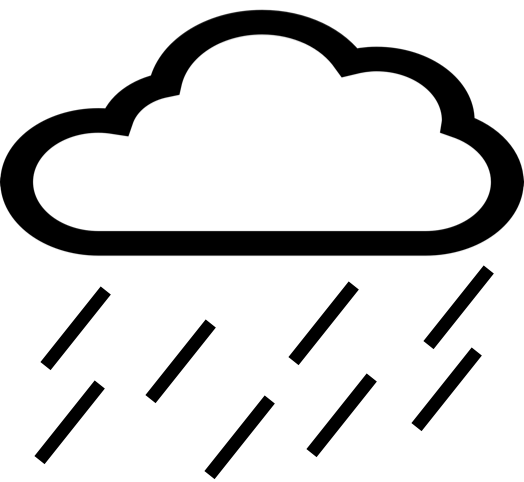}}}
		{\raisebox{-2pt}{\includegraphics[height=2ex]{thunderstorm.png}}}
		{\includegraphics[height=1.5ex]{thunderstorm.png}}
		{\includegraphics[height=1ex]{thunderstorm.png}}
}}

\newcommand{\badthunderstorm}{\ensuremath{%
		\mathchoice{\raisebox{-2pt}{\includegraphics[height=2ex]{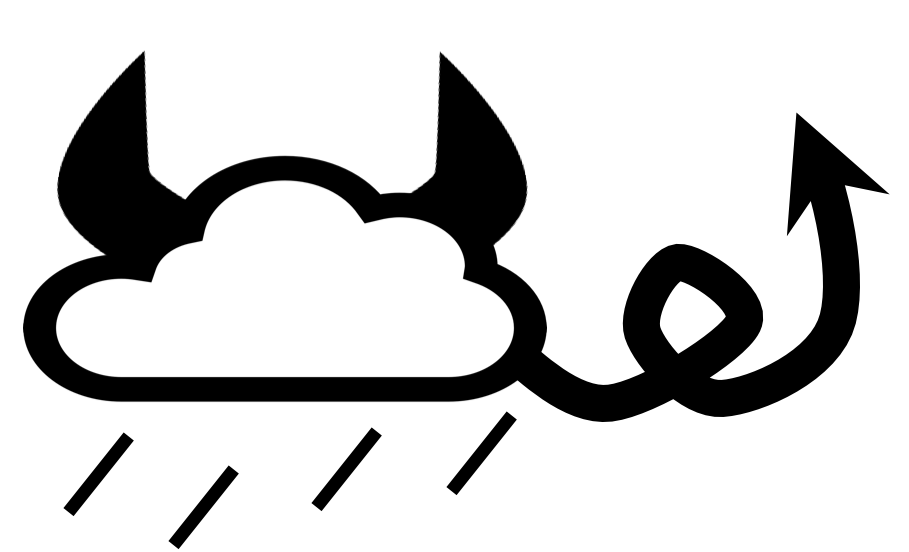}}}
		{\raisebox{-2pt}{\includegraphics[height=2ex]{badthunderstorm.png}}}
		{\includegraphics[height=1.5ex]{badthunderstorm.png}}
		{\includegraphics[height=1ex]{badthunderstorm.png}}
}}

\newcommand{\F}{{\mathbb F}} 
 
\newcommand{\Z}{\mathbb{Z}}

\newcommand{\C}{\mathbb{C}}

\newcommand{\vv}{{\mathbf{v}}}

\newcommand{\Am}{{\mathbf{A}}}
\newcommand{\Bm}{{\mathbf{B}}}
\newcommand{\Cm}{{\mathbf{C}}}

\newcommand{\Em}{{\mathbf{E}}}

\newcommand{\Mm}{{\mathbf{M}}}

\newcommand{\Tr}{{\sf Tr}}

\newcommand{\As}{{\mathcal{A}}}

\newcommand{\Cs}{{\mathcal{C}}}
\newcommand{\Ds}{{\mathcal{D}}}

\newcommand{\Fs}{{\mathcal{F}}}

\newcommand{\Ms}{{\mathcal{M}}}

\newcommand{\Ss}{{\mathcal{S}}}

\newcommand{\adv}{{\As}}

\newcommand{\advlightning}{\text{\resizebox{1em}{1em}{\badthunderstorm}}}

\newcommand{\genbolt}{\text{\raisebox{0pt}{\thunderstorm}}}
\newcommand{\verbolt}{{\sf Ver}}
\newcommand{\serialnumber}{{s}}
\newcommand{\qsetup}{{\sf SetupQL}}
\newcommand{\qlightning}[1][{}]{{|\text{\Lightning}_{#1}\rangle}}
\newcommand{\qlightningp}[1][{}]{{|\text{\Lightning}'_{#1}\rangle}}

\newcommand{\genmoney}{{\sf Gen}}
\newcommand{\vermoney}{{\sf Ver}}

\newcommand{\owf}{{\sf OWF}}
\newcommand{\iO}{{\sf iO}}
\newcommand{\shO}{{\sf shO}}
\newcommand{\Samp}{{\sf Samp}}

\newcommand{\poly}{{\sf poly}}
\newcommand{\negl}{{\sf negl}}

\newcommand{\sk}{{\sf sk}}
\newcommand{\pk}{{\sf pk}}
\newcommand{\aux}{{\sf aux}}
\newcommand{\gen}{{\sf Gen}}
\newcommand{\ver}{{\sf Ver}}
\newcommand{\sign}{{\sf Sign}}

\newcommand{\comm}{{\sf Commit}}
\newcommand{\reveal}{{\sf Reveal}}
\newcommand{\state}{{\sf State}}

\def\E{\mathop{\mathbb E}}

\newcommand{\wbzexp}{{\tt{W\text{-}BZ\text{-}Exp}}}
\newcommand{\sbzexp}{{\tt{S\text{-}BZ\text{-}Exp}}}
\newcommand{\gyzexp}{{\tt{GYZ\text{-}Exp}}}

\newcommand{\collapseexp}{{\tt{Collapse\text{-}Exp}}}

\newcommand{\ignore}[1]{}

%% file: intro.tex
\label{sec:intro}

Unlike classical bits, which can be copied ad nauseum, quantum bits --- called qubits --- cannot in general be copied, as a result of the Quantum No-Cloning Theorem.  No-cloning has various negative implications to the handling of quantum information; for example it implies that classical error correction cannot be applied to quantum states, and that it is impossible to transmit a quantum state over a classical channel.  On the flip side, no-cloning has tremendous potential for cryptographic purposes, where the adversary is prevented from various strategies that involve copying.  For example, Wiesner~\cite{Wiesner83} shows that if a quantum state is used as a banknote, no-cloning means that an adversary cannot duplicate the note.  This is clearly impossible with classical bits.  Wiesner's idea can also be seen as the starting point for quantum key distribution~\cite{BenBra84}, which can be used to securely exchange keys over a public channel, even against computationally unbounded eavesdropping adversaries.

In this work, we investigate no-cloning in the presence of computationally bounded adversaries, and it's implications to cryptography.  To motivate this discussion, consider the following two important applications:
\begin{itemize}
	\item A public key quantum money scheme allows anyone to verify banknotes.  This remedies a key limitation of Wiesner's scheme, which requires sending the banknote back to the mint for verification.  The mint has a \emph{secret} classical description of the banknote which it can use to verify; if this description is made public, then the scheme is completely broken.  Requiring the mint for verification represents an obvious logistical hurdle.  In contrast, a public key quantum money scheme can be verified locally without the mint's involvement.  Yet, even with the ability to verify a banknote, it is impossible for anyone (save the mint) to create new notes.
	\item Many cryptographic settings such as multiparty computation require a random string to be created by a trusted party during a set up phase.  But what if the randomness creator is not trusted?  One would still hope for some way to verify that the strings it produces are still random, or at least have some amount of (min-)entropy.  At a minimum, one would hope for a guarantee that their string is different from any previous or future string that will be generated for anyone else.  Classically, these goals are impossible.  But quantumly, one may hope to create proofs that are unclonable, so that only a single user can possibly ever receive a valid proof for a particular string.
\end{itemize}

Notice that in both settings above, a computationally unbounded adversary can always break the scheme.  For public key quantum money, the following attack produces a valid banknote from scratch in exponential-time: generate a random candidate quantum money state and apply the verification procedure.  If it accepts, output the state; otherwise try again.  Similarly, in the verifiable randomness setting, an exponential-time adversary can always run the randomness generating procedure until it gets two copies of the same random string, along with two valid proofs for that string.  Then it can give the same string (but different valid proofs) to two different users.  With the current state of knowledge of complexity theory, achieving security against a computationally bounded adversary means computational assumptions are required; in particular, both scenarios imply at a minimum one-way functions secure against quantum adversaries.  

\medskip

Combining no-cloning with computational assumptions is a subtle task.  A typical computational assumption posits that a \emph{classical} problem (that is, the inputs and outputs of the problem are classical) is computationally hard, such as solving general NP problems or finding short vectors in lattices.  In the quantum regime, we ask that the classical problem is hard \emph{even for quantum computers.}  However, it is still desirable for the assumption to involve a classical problem, rather than a quantum problem (quantum states as inputs and outputs).  This is for several reasons.  For one, we want a violation of the assumption to lead to a mathematically interesting result, and this seems much more likely for classical problems.  Furthermore, it is much harder for the research community to study and analyze a quantum assumption, since it will be hard to isolate the features of the problem that make it hard.  Therefore, an important goal in quantum cryptography is 

\begin{center}{\it Combine no-cloning and computational assumptions about\\ classical problems to obtain no-cloning-with-verification.}
\end{center}

In addition to the underlying assumption being classical, it would ideally also be one that has been previously studied by cryptographers, and ideally used in other cryptographic contexts.  This would give the strongest possible evidence that the assumption, and hence application, are secure.

\medskip

For now, we focus on the setting of public key quantum money.  Constructing such quantum money from a classical hardness assumption is a surprisingly difficult task.  One barrier is the following.  Security would be proved by reduction, and algorithm that interacts with a supposed quantum money adversary and acts as an adversary for the the underlying \emph{classical} computational assumption.  Note that the adversary expects as input a valid banknote, which the reduction must supply.  Then it appears the reduction should somehow use the adversary's forgery to break the computational assumption.  But if the reduction can generate a single valid banknote, there is nothing preventing it from generating a second --- recall that the underlying assumption is classical, so we cannot rely on the assumption to provide us with an un-clonable state.  Therefore, if the reduction works, it would appear that the reduction can create two banknotes for itself, in which case it can break the underlying assumption without the aid of the adversary.  This would imply that the underlying assumption is in fact false.

The above difficulties become even more apparent when considering the known public key quantum money schemes.  The first proposed scheme by Aaronson~\cite{CCC:Aaronson09} had no security proof, and was subsequently broken by Lutomirski et al.~\cite{ITCS:LAFGKH10}.  The next proposed scheme by Farhi et al.~\cite{ITCS:FGHLS12} also has no security proof, though this scheme still remains unbroken.  However, the scheme is complicated, and it is unclear which quantum states are accepted by the verification procedure; it might be that there are dishonest banknotes that are both easy to construct, but are still accepted by the verification procedure.

Finally, the third candidate by Aaronson and Christiano~\cite{STOC:AarChr12} actually \emph{does} prove security using a classical computational problem.  However, in order to circumvent the barrier discussed above, the classical problem has a highly non-standard format.  They observe that a polynomial-time algorithm can, by random guessing, produce a valid banknote with some exponentially-small probability $p$, while random guessing can only produce two valid banknotes with probability $p^2$.  Therefore, their reduction first generates a valid banknote with probability $p$, runs the adversary on the banknote, and then uses the adversary's forgery to increase its success probability for some task.  This reduction strategy requires a very carefully crafted assumption, where it is assumed hard to solve a particular problem in polynomial time with exponentially-small probability $p$, even though it can easily be solved with probability $p^2$.  

In contrast, typical assumptions in cryptography involve polynomial-time algorithms and \emph{inverse-polynomial} success probabilities, rather than exponential.  (Sub)exponential hardness assumptions are sometimes made, but even then the assumptions are usually closed under polynomial changes in adversary running times or success probabilities, and therefore make no distinction between $p$ and $p^2$.  In addition to the flavor of assumption being highly non-standard, Aaronson and Christiano's assumption --- as well as their scheme --- have been subsequently broken~\cite{PKC:PenFauPer15,Aaronson16}. 

\medskip

Turning to the verifiable randomness setting, things appear even more difficult.  Indeed, our requirements for verifiable randomness imply an even stronger version of computational no-cloning: an adversary should not be able to copy a state, even if it can verify the state, and \emph{even if it devised the original state itself}.  Indeed, without such a restriction, an adversary may be able to come up with a dishonest proof of randomness, perhaps by deviating from the proper proof generating procedure, that it \emph{can} clone arbitrarily many times.  Therefore, a fascinating objective is to

\begin{center}{\it Obtain a no-cloning theorem, even for settings where the adversary\\ controls the entire process for generating the original state.}
\end{center}

\subsection{This Work: Strong Variants of No-Cloning}

In this work, we study strong variants of quantum no-cloning, in particular public key quantum money, and uncover relationships between no-cloning and various cryptographic applications.

\subsubsection{Quantum Lightning Never Strikes the Same State Twice}

The old adage about lightning is of course false, but the idea nonetheless captures some of the features we would like for the verified randomness setting discussed above.  Suppose a magical randomness generator could go out into a thunderstorm, and ``freeze'' and ``capture'' lightning bolts as they strike.  Every lightning bolt will be different.  The randomness generator then somehow extracts a fingerprint or serial number from the frozen lightning bolt (say, hashing the image of the bolt from a particular direction).  The serial number will serve as the random string, and the frozen lightning bolt will be the proof of randomness; since every bolt is different, this ensures that the bolts, and hence serial numbers, have some amount of entropy.

Of course, it may be that there are other ways to create lightning other than walking out into a thunderstorm (Tesla coils come to mind).  We therefore would like that, no matter how the lightning is generated, be it from thunderstorms or in a carefully controlled laboratory environment, every bolt has a unique fingerprint/serial number. 

We seek a complexity-theoretic version of this magical frozen lightning object, namely a phenomenon which guarantees different outcomes every time, no matter how the phenomenon is generated.  We will necessarily rely on quantum no-cloning --- since in principle a classical phenomenon can be replicated by starting with the same initial conditions --- and hence we call our notion \emph{quantum lightning}.  Quantum lightning, roughly, is a strengthening of public key quantum money where the procedure to generate new banknotes itself is public, allowing anyone to generate banknotes.  Nevertheless, it is impossible for an adversary to construct two notes with the same serial number.  This is a surprising and counter-intuitive property, as the adversary knows how to generate banknotes, and moreover has full control over how it does so; in particular it can deviate from the generation procedure any way it wants, as long as it is computationally efficient.  Nonetheless, it cannot devise a malicious note generation procedure that allows it to construct the same note twice.  This concept of quantum money can be seen as a formalization of the concept of ``collision-free'' public key quantum money due to Lutomirski et al.~\cite{ITCS:LAFGKH10}.  

\medskip

Slightly more precisely, a quantum lightning protocol consists of two efficient (quantum) algorithms.  The first is a bolt generation procedure, or storm, $\genbolt$, which generates a quantum state $\qlightning$ on each invocation.  The second algorithm, $\verbolt$, meanwhile verifies bolts as valid and also extracts a fingerprint/serial number of the bolt.  For correctness, we require that (1) $\verbolt$ always accepts bolts produced by $\genbolt$, (2) it does not perturb valid bolts, and (3) that it will always output the same serial number on a given bolt.

For security, we require the following: it is computationally infeasible to produce two bolts $\qlightning[0]$ and $\qlightning[1]$ such that $\verbolt$ accepts both and outputs identical serial numbers.  This is true for even for \emph{adversarial} storms $\advlightning$, even those that depart from $\genbolt$ or produce entangled bolts, so long as $\advlightning$ is computationally efficient.

\paragraph{Applications.} Quantum lightning as described has several interesting applications:
\begin{itemize}
	\item {\bf Quantum money.}  Quantum lightning easily gives quantum money.  A banknote is just a bolt, with the associated serial number signed by the bank using an arbitrary classical signature scheme.  Any banknote forgery must either forge the bank's signature, or must produce two bolts with the same serial number, violating quantum lightning security.
	\item {\bf Verifiable min-entropy}.  Quantum lightning also gives a way to generate random strings along with a proof that the string is random, or at least has min-entropy.  To see this, consider an adversarial bolt generation procedure that produces bolts such that the associated serial number has low min-entropy.  Then by running this procedure several times, one will eventually obtain in polynomial time two bolts with the same serial number, violating security.
	
	Therefore, to generate a verifiable random string, generate a new bolt using $\genbolt$.  The string is the bolt's serial number, and $\genbolt$ serves as a proof of min-entropy, which is verified using $\verbolt$.
	
	\item {\bf Decentralized Currency.} Finally, quantum lightning yields a simple new construction of totally decentralized digital currency.  Coins are just bolts, except the serial number must hash to a string that begins with a certain number of 0's.  Anyone can produce coins by generating bolts until the hash begins with enough 0's.  Moreover, verification is just $\verbolt$, and does not require any interaction or coordination with other users of the system.  This is an advantage over classical cryptocurrencies such as BitCoin, which require a large public and dynamic ledger, and requires a pool of miners to verify transactions.  Our protocol does have significant limitations relative to classical cryptocurrencies, which likely make it only a toy object.  We hope that further developments will yield a scheme that overcomes these limitations.
\end{itemize}

\subsubsection{Connections to Post-quantum Security}

One simple folklore way to construct a state that can only be constructed once but never a second time is to use a collision-resistant hash function $H$.  First, generate a uniform superposition of inputs.  Then apply the $H$ in superposition, and measure the result $y$.  The state collapses to the superposition $|\psi_y\rangle$ of all pre-images $x$ of $y$.  

Notice that, while it is easy to sample states $|\psi_y\rangle$, it is impossible to sample two copies of the same $|\psi_y\rangle$.  Indeed, given two copies of $|\psi_y\rangle$, simply measure both copies.  Since these are superpositions over many inputs, each state will likely yield a different $x$.  The two $x$'s obtained are both pre-images of the same $y$, and therefore constitute a collision for $H$.

The above idea does not yet yield quantum lightning.  For verification, one can hash the state to get the serial number $y$, but this alone is insufficient.  For example, an adversarial storm can simply choose a random string $x$, and output $|x\rangle$ twice as its two copies of the same state.  Of course, $|x\rangle$ is not equal to $|\psi_y\rangle$ for any $y$.  However, the verification procedure just described does not distinguish between these two states.  

What one needs therefore is mechanism to distinguish a random $|x\rangle$ from a random $|\psi_y\rangle$.  Interestingly, as observed by Unruh~\cite{EC:Unruh16}, this is exactly \emph{the opposite} what one would normally want from a hash function.  Consider the usual way of building a computationally binding commitment from a collision resistant hash function: to commit to a message $m$, choose a random $r$ and output $H(m,r)$.  Classically, this is computationally binding by the collision resistance of $H$: if an adversary can open the commitment to two different values, this immediately yields a collision for $H$.  Unruh~\cite{EC:Unruh16} shows in the quantum setting, collision resistance --- even against quantum adversaries --- is not enough.  Indeed, he shows that for certain hash functions $H$ it may be possible for the adversary to produce a commitment, and only afterward decide on the committed value.  Essentially, the adversary constructs a superposition of pre-images $|\psi_y\rangle$ as above, and then uses particular properties of $H$ to perturb $|\psi_y\rangle$ so that it becomes a different superposition of pre-images of $y$.   Then one simply de-commits to any message by first modifying the superposition and then measuring.  This does not violate the collision-resistance of $H$: since the adversary cannot copy $|\psi_y\rangle$, the adversary can only ever perform this procedure once and obtain only a single de-commitment.  

To overcome this potential limitation, Unruh defines a notion of \emph{collapsing} hash functions.  Roughly, these are hash functions for which $|x\rangle$ and $|\psi_y\rangle$ are \emph{indistinguishable}.  Using such hash functions to build commitments, one obtains \emph{collapse-binding} commitments, for which the attack above is impossible.  Finally, he shows that a random oracle is collapse binding.

More generally, an implicit assumption in many classical settings is that, if an adversary can modify one value into another, then it can produce both the original and modified value simultaneously.  For example, in a commitment scheme, if a classical adversary can de-commit to both 0 or 1, it can then also simultaneously de-commit to both 0 and 1 by first de-committing to 0, and then re-winding and de-committing to 1.  Thus it is natural classically to require that it is impossible to simultaneously produce de-commitments to both 0 and 1.  Similarly, for signatures, if an adversary can modify a signed message $m_0$ into a signed message $m_1$, then it can simultaneously produce two signed messages $m_0,m_1$.  This inspires the Boneh-Zhandry~\cite{EC:BonZha13,C:BonZha13} definition of security for signatures in the presence of quantum adversaries, which says that after seeing a (superposition of) signed messages, the adversary cannot produce two signed messages.

However, a true quantum adversary may be able, for some schemes, to set things up so that it can modify a (superposition) of values into one of many possibilities, but still only be able to ever produce a single value.  For example, it many be that an adversary sees a superposition of signed messages that always begin with 0, but somehow modifies the superposition to obtain a signed message that begins with a 1.  This limitation for signatures was observed by Garg, Yuen, and Zhandry~\cite{C:GarYueZha17}, who then give a much stronger notion to fix this issue\footnote{Garg et al. only actually discuss message authentication codes, but the same idea applies to signatures}.

\medskip

Inspired by the above observations, we formulate a series of win-win results for quantum lightning and quantum money.  In particular, we show, roughly,

\begin{theorem}[informal] If $H$ is a hash function that is collision resistant against quantum adversaries, then either (1) $H$ is collapsing or (2) it can be used to build quantum lightning without any additional computational assumptions.\footnote{Technically, there is a slight gap due to the difference between \emph{non-negligible} and \emph{inverse polynomial}.  Essentially what we show is that the theorem holds for fixed values of the security parameter, but whether (1) or (2) happens may vary across different security parameters.}
\end{theorem}
The construction of quantum lightning is inspired by the outline above.  One difficulty is that above we needed a perfect distinguisher, whereas a collapsing adversary may only have a non-negligible advantage.  To obtain an actual quantum lightning scheme, we need to repeat the scheme in parallel many times to boost the distinguish advantage to essentially perfect.  Still, defining verification so that we can prove security is a non-trivial task.  Indeed, while it is possible to define a verification procedure that accepts valid bolts, it is much harder to analyze what sorts of invalid bolts might be accepted by the verification procedure.  For example, it may be that a certain sparse distribution on $|x\rangle$ is actually accepted by the verification procedure, even though a random $|x\rangle$ is almost always rejected.  Using a careful argument, we show nonetheless how to verify and prove security.  We also show that

\begin{theorem}[informal] Any \emph{non-interactive} commitment scheme that is computationally binding against quantum adversaries is either collapse-binding, or it can be used to build quantum lightning without any additional computational assumptions.
\end{theorem}

The above theorem crucially relies on the commitment scheme being non-interactive: the serial number of the bolt is the sender's single message, along with his private quantum state.  If the commitment scheme is not collapse-binding, the sender's private state can be verified to be in superposition.  If a adversary produces two identical bolts, these bolts can be measured to obtain two openings, violating computational binding.  In contrast, in the case of interactive commitments, the bolt should be expanded to the transcript of the interaction between the sender and receiver.  Unfortunately, for quantum lightning security, the transcript is generated by an adversary, who can deviate from the honest receiver's protocol.  Since the commitment scheme is only binding when the receiver is run honestly, we cannot prove security in this setting.

Instead, we consider the weaker goal of constructing public key quantum money.  Here, since the mint produces bolts, the original bolt is honestly generated.  The mint then signs the transcript using a standard signature scheme (which can be built from one-way functions, and hence implied by commitments).  If the adversary duplicates this banknote, it is duplicating an honest commitment transcript, but the note can be measured to obtain two different openings, breaking computational binding.  This gives us the following:

\begin{theorem}[informal] Any interactive commitment scheme that is computationally binding against quantum adversaries is either collapse-binding, or it can be used to build public key quantum money without any additional computational assumptions.
\end{theorem}

\noindent Finally, we extend these ideas to a win-win result for quantum money and digital signatures:

\begin{theorem}[informal] Any one-time signature scheme that is Boneh-Zhandry secure is either Garg-Yuen-Zhandry secure, or it can be used to build public key quantum money without any additional computational assumptions.
\end{theorem}

Given the difficulty of constructing public key quantum money (let alone quantum lightning), the above results suggest that most natural constructions of collision resistant hash functions are likely already collapsing, with analogous statements for commitment schemes and signatures.  If they surprisingly turn out to not meet the stronger quantum notions, then we would immediately obtain a construction of public key quantum money from simple tools.

\medskip

Notice that using our win-win results give a potential route toward proving the security of quantum money/lightning in a way that avoids the barrier discussed above.  Consider building quantum money from quantum lightning, and in turn building quantum lightning from a collision-resistant non-collapsing hash function.  Recall that a banknote is a bolt, together with the mint's signature on the bolt's serial number.  A quantum money adversary either (1) duplicates a bolt to yield two bolts with the same serial number (and hence same signature), or (2) produces a second bolt with a different serial number, as well as a forged signature on that serial number.  Notice that (2) is impossible simply by the unforgeability of the mint's signature.  Meanwhile, in proving that (1) is impossible, our reduction actually \emph{can} produce arbitrary quantum money states (for this step, we assume the reduction is given the signing key).  The key is that the reduction on its own \emph{cannot} produce the same quantum money state twice, but it \emph{can} do so using the adversary's cloning abilities, allowing it to break the underlying hard problem.

\subsubsection{Constructing Quantum Lightning}

We now turn to actually building quantum lightning, and hence quantum money.  Following our win-win results, we would like a \emph{non-collapsing} collision-resistant hash function.  Unfortunately, Unruh's counterexample does not yield an explicit construction.  Instead, he builds on techniques of~\cite{FOCS:AmbRosUnr14} to give a hash function relative to a \emph{quantum} oracle\footnote{that is, the oracle itself performs quantum operations}.  As it is currently unknown how to obfuscate quantum oracles with a meaningful notion of security, this does not give even a candidate construction of quantum lightning.  Instead, we focus on specific standard-model constructions of hash functions.  Finding suitable hash functions is surprisingly challenging; we were only able to find a single candidate, and leave finding additional candidates as a challenging open problem.

Our construction is based on low-degree hash functions, which were previously studied by Ding and Yang~\cite{DingYang08} and Applebaum et al.~\cite{ITCS:AHIKV17} with the goal of constructing very efficient hash functions.  Specifically, we will use hash functions defined by a random set of degree-2 polynomials over $\F_2$.  That is, the hash function is specified by a set of polynomials $P_1,\dots,P_n$ in $m$ variables over $\F_2$, where each $P_i$ is of degree 2.  The output of $H(x)$ is just the set of $i$ outputs $P_i(x)$.  

As we will see below, such hash functions are not collision resistant, and therefore are not immediately applicable to our win-win results.  Nonetheless, these hash functions will serve as a useful starting point.  In particular, we show that if $m\approx n^2$, then $H$ is not collapsing.  That is, it is possible to distinguish $|\psi_y\rangle$ (the uniform superposition over pre-images of $y$) from $|x\rangle$, for a random $x,y$.  If the $P_i$'s were linear even for $m=\Omega(n)$, this would be trivial by taking the quantum Fourier transform (which over $\F_2$ is just the Hadamard gate).  Since the $P_i$'s have degree 2, our distinguisher is more complicated.  Instead, we observe that if we perform the Hadamard gate to just one input qubit and measure, the result is a derivative of some unknown linear combination of the $P_i$s; the derivative therefore effectively reduces the problem to a linear one, which we can solve more easily.  From this measurement, we obtain a single linear equation in $n$ unknowns, whose coefficients are determined by the remaining $m-1$ inputs bits.  

We cannot simply repeat the process to get more linear equations, since the Hadamard gate introduced a phase term.  Moreover, if we record the coefficients of the linear equation, this amounts to a measurement on the remaining $m-1$ inputs bits, further perturbing the state.  Nevertheless, we show how to correct the phase term and the measurement by performing a linear transformation to the $m-1$ remaining inputs.  This process consumes $n$ inputs bits, resulting in a system on $m-n-1$ qubits.  However, now the state is ready to generate another linear equation in the same $n$ unknowns.  We then repeat approximately $n$ times, and the resulting system of equations has a solution if and only if the original state was $|\psi_y\rangle$ for some $y$.  Overall, this distinguisher requires $m\approx n^2$ input bits.  

Unfortunately, degree-2 hash functions are not collision resistant, as shown by Ding and Yang~\cite{DingYang08} and Applebaum et al.~\cite{ITCS:AHIKV17}.  This is because the equations $P_i(x_0)=P_i(x_1)$ can be linearized and then solved using linear algebra.  This attack exists even if $m=n+1$, and therefore certainly applies to our setting where $m\approx n^2$.  Moreover, we show how to extend the attack so that if $m\approx rn$, it is possible to construct a dimension-$r$ affine space of colliding inputs, thus giving $2^r$ colliding inputs.  If we insist on the colliding inputs having no affine relations, then we show that it is still possible to construct $r+1$ colliding inputs.  This means that, for our non-collapsing hash function above, not only is it non-collision resistant, but it is possible to find many colliding inputs.  Even worse, we can use our attack to generate $r+1$ identical copies of $|\psi_y\rangle$ for the same $y$, meaning our scheme so far is certainly not a quantum lightning scheme.

\smallskip

Despite the above attacks, Applebaum et al. conjecture that random degree-2 polynomials are still one-way, and even second-pre-image resistant.  This is because the above attacks crucially have little control over the output $y$ or the particular colliding inputs that are generated.  We conjecture moreover that it is impossible to devise $2(r+1)$ colliding inputs that have no affine relations.  This conjecture seems plausible in light of known attacks; note that the best-known attacks cannot even generate $r+2$ non-affine colliding inputs.  

Using this conjecture, we obtain a quantum lightning scheme as follows.  Our bolt is $\qlightning=|\psi_y\rangle^{\otimes (r+1)}$ for some $y$, which we generate using the attack above.  For verification, if we set $r\approx n$, we can use our distinguisher on each copy of $|\psi_y\rangle$ to verify that it has the proper form.  We also verify that each of the $r$ states have the same hash $y$ under $H$.  The serial number for $\qlightning$ is $y$.  

If an adversary constructs two states $\qlightning[0],\qlightning[1]$ which both pass verification and have the same serial number $y$, we show that the joint state \emph{after} verification must be precisely $|\psi_y\rangle^{\otimes 2(r+1)}$.  By measuring this state, we obtain $2(r+1)$ random pre-images.  With overwhelming probability, these colliding inputs will have no affine relationships.  Therefore, any quantum lightning adversary violates our conjecture.

\medskip

We note that our quantum lightning scheme is not an instance of a collision-resistant non-collapsing hash function.  Nonetheless, we show that it can be modified to obtain such a function.  Security will still be based on the same exact computational assumption, though the derived hash function will no longer be a degree-2 hash function.

Thus we obtain the first construction of quantum lightning, and hence public key quantum money, based on a plausible conjecture about the quantum hardness of a \emph{classical} computational problem.  Moreover our problem has to some extent been previously studied in the literature.  We also give the first plausible \emph{standard model} construction of a hash function that is not collapsing.  In the process, we establish a general approach to constructing quantum lightning/money, by constructing collision resistant hash functions that are not collapsing.  We leave as an interesting open question whether such hash functions can be constructed from more mainstream assumptions, such as the hardness of lattice problems or on generic assumptions.

\subsection{Quantum Money From Obfuscation}

Finally, we consider the simpler task of constructing public key quantum money.  One possibility is based on Aaronson and Christiano's broken scheme~\cite{STOC:AarChr12}.  In their scheme, a quantum banknote $|\$\rangle$ is a uniform superposition over some subspace $S$, that is known only to the bank.  The quantum Fourier transform of such a state is the uniform superposition over the dual subspace $S^\bot$.  This gives a simple way to check the banknote: test if $|\$\rangle$ lies in $S$, and whether it's Fourier transform lies in $S^\bot$.  Aaronson and Christiano show that the only state which can pass verification is $|\$\rangle$.  

To make this scheme public key, one gives out a mechanism to test for membership in $S$ and $S^\bot$, without actually revealing $S,S^\bot$.  This essentially means obfuscating the functions that decide membership.  Aaronson and Christiano's scheme can be seen as a candidate obfuscator for subspaces.  While unfortunately their obfuscator has since been broken, one may hope to instantiate their scheme using recent advances in general-purpose program obfuscation, specifically indistinguishability obfuscation (iO)~\cite{C:BGIRSVY01,FOCS:GGHRSW13}.  

On the positive side, Aaronson and Christiano show that their scheme is secure if the subspaces are provided as quantum-accessible black boxes, giving hope that some obfuscation of the subspaces will work.  Unfortunately, proving security relative to iO appears a difficult task.  One limitation is the barrier discussed above, that any reduction must be able to produce a valid banknote, which means it can also produce two banknotes.  Yet at the same time, it somehow has to use the adversary's forgery (a second banknote) to break the iO scheme.  Note that this situation is different from the quantum lightning setting, where there were many valid states, and no process could generate the same state twice.  Here, there is a single valid state (the state $|\$\rangle$), and it would appear the reduction must be able to construct this precise state exactly once, but not twice.  Such a reduction would clearly be impossible.  As discussed above Aaronson and Christiano circumvent this issue by using a non-standard type of assumption; their technique is not relevant for standard definitions of iO.

\medskip

Our solution is to separate the proof into two phases.  In the first, we change the spaces obfuscated from $S,S^\bot$ to $T_0,T_1$, where $T_0$ is a random unknown subspace containing $S$, and $T_1$ is a unknown random subspace containing $S^\bot$.  This modification can be proved undetectable using a weak form of obfuscation we define, called subspace-hiding obfuscation, which in turn is implied by iO.  Note that in this step, we even allow the reduction to know $S$ (but not $T_0,T_1$), so it can produce as many copies of $|\$\rangle$ as it would like to feed to the adversary.  The reduction does not care about the adversary's forgery directly, only whether or not the adversary successfully forges.  If the adversary forges when given obfuscations of $S,S^\bot$, it must also forge under $T_0,T_1$, else it can distinguish the two cases and hence break the obfuscation.  By using the adversary in this way, we avoid the apparent difficulties above.

In the next step, we notice that, conditioned on $T_0,T_1$, the space $S$ is a random subspace between $T_1^\bot$ and $T_0$.  Thus conditioned on $T_0,T_1$, the adversary clones a state $|\$\rangle$ defined by a random subspace $S$ between $T_1^\bot$ and $T_0$.  The number of possible $S$ is much larger than the dimension of the state $|\$\rangle$, so in particular the states cannot be orthogonal.  Thus, by no-cloning, duplication is impossible.  We need to be careful however, since we want to rule out adversaries that forge with even very low success probabilities.  To do so, we need to precisely quantify the no-cloning theorem, which we do.  We believe our new no-cloning theorem may be of independent interest.  We note that when applying no-cloning, we do not rely on the secrecy of $T_0,T_1$, but only that $S$ is hidden.  Intuitively, there are exponentially many more $S$'s between $T_0,T_1$ than the dimension of the space $|\$\rangle$ belongs to, so no-cloning implies that a forger has negligible success probability. Thus we reach a contradiction, showing that the original adversary could not exist.

\medskip

We also show how to view Aaronson and Christiano's scheme as a signature scheme; we show that the signature scheme satisfies the Boneh-Zhandry definition, but not the strong Garg-Yuen-Zhandry notion.  Thus, we can view Aaronson and Christiano's scheme as an instance of our win-win results, and moreover provide the first separation between the two security notions for signatures.

We note that our result potentially relies on a much weaker notion of obfuscation that full iO, giving hope that security can be based on weaker assumptions.  For example, an intriguing open question is whether or not recent constructions of obfuscation for certain evasive functions~\cite{EPRINT:WicZir17,EPRINT:GoyKopWat17} based on LWE can be used to instantiate our notion of subspace hiding obfuscation.  This gives another route toward building quantum money from hard lattice problems.  This is particularly important at the present time, where new quantum attacks have called into question the security of full-fledged iO in the quantum setting (see below for a discussion).  

Finally, we note that independently of whether iO exists in the quantum setting, black box oracles do provide subspace hiding, a fact which can be based easily on the quantum lower bounds for unstructured search~\cite{BBBV97}.  With this insight, our proof strategy can be used to give a simplified analysis of Aaronson and Christiano's black-box scheme.  Their proof relied on developing a new type of adversary method, called the inner-product adversary method.  Instead, we rely simply on the lower bound for unstructured search plus our quantitative no cloning theorem.

\subsection{Related Works}

\paragraph{Quantum Money.} Lutomirski~\cite{Lutomirski10} shows another weakness of Wiesner's scheme: a merchant, who is allowed to interact with the mint for verification, can use the verification oracle to break the scheme and forge new currency. Public key quantum money is necessarily secure against adversaries with a verification oracle, since the adversary can implement the verification oracle for itself.  Several alternative solutions to the limitations of Wiesner's scheme have been proposed~\cite{MosSte10,Gavinsky11}, though the ``ideal'' solution still remains public key quantum money.

\paragraph{Randomness Expansion.}  Colbeck~\cite{Colbeck09} proposed the idea of a classical experimenter, interacting with several potentially untrustworthy quantum devices, can expand a small random seed into a certifiably random longer seed.  Here, a crucial assumption is that the devices cannot communicate and must obey the laws of quantum mechanics; no other assumption about the devices is made.  The application of quantum lightning to verified randomness has a similar spirit, though the requirements are quite different.  Randomness expansion requires multiple non-communicating devices, but the experimenter can be classical and the devices can have unbounded computational power; in contrast quantum lightning involves only a single device, but the device must be computationally bounded, and the experimenter must perform quantum operations.  We note that a quantum experimenter can generate a random string for free; the purpose of verifiable entropy in this case is simply to prove to another individual that the coins you generated were indeed random.

\paragraph{Obfuscation and Multilinear Maps.}  There is a vast body of literature on strong notions of obfuscation, starting with the definitional work of Barak et al.~\cite{C:BGIRSVY01}.  Garg et al.~\cite{FOCS:GGHRSW13} propose the first obfuscator plausibly meeting the strong notion of iO, based on cryptographic multilinear maps~\cite{EC:GarGenHal13,C:CorLepTib13,TCC:GenGorHal15}.  Unfortunately, there have been numerous attacks on multilinear maps, which we do not fully elaborate on here.  Importantly, all current obfuscators are subject to very strong quantum attacks~\cite{EC:CDPR16,C:AlbBaiDuc16,CJL16,EC:CheGenHal17}, casting doubt on their quantum security.  However, there has been some success in transforming applications of obfuscation to be secure under assumptions on lattices~\cite{ITCS:BVWW16,EPRINT:WicZir17,EPRINT:GoyKopWat17}, which are widely believed to be quantum hard.   We therefore think it plausible that subspace-hiding obfuscation, which is all we need for this work, can be based on similar lattice problems.  Nonetheless, obfuscation is a very active area of research, and we believe that obfuscation secure against quantum attacks will likely be discovered in the near future.

\paragraph{Computational No-cloning.} We note that computational assumptions and no-cloning have been combined in other contexts, such as Unruh's revocable time-released encryption~\cite{EC:Unruh14}.  We note however, that these settings do not involve verification, the central theme of this work.

%% file: prelim.tex
\label{sec:prelim}

\subsection{Notations}

Throughout this paper, we will let $\lambda$ be a security parameter.  When inputted into an algorithm, $\lambda$ will be represented in unary.

A function $\epsilon(\lambda)$ is \emph{negligible} if for any inverse polynomial $1/p(\lambda)$, $\epsilon(\lambda)<1/p(\lambda)$ for sufficiently large $\lambda$.  A function is \emph{non-negligible} if it is not negligible, that is there exists an inverse polynomial $1/p(\lambda)$ such that $\epsilon(\lambda)\geq 1/p(\lambda)$ infinitely often.

\subsection{Quantum Computation}

A quantum system $Q$ is defined over a finite set $B$ of classical states.  We will generally consider $B=\{0,1\}^n$.  A {\bf pure} state over $Q$ is an $L_2$-normalized vector in $\C^{|B|}$, which assigns a (complex) weight to each element in $B$.  Thus the set of pure states forms a complex Hilbert space.  A {\bf qubit} is a quantum system defined over $B=\{0,1\}$.  Given a quantum system $Q_0$ over $B_0$ and a quantum system $Q_1$ over $B_1$, we can define the product system $Q=Q_0\times Q_1$ over $B=B_0\times B_1=\{(b_0,b_1):b_0\in B_0,b_1\in B_1\}$.  Given a state $v_0\in Q_0$ and $v_1\in Q_1$, we define the product state $v_0\otimes v_1$ in the natural way.  An $n$-qubit system is then $Q=Q_0^{\oplus n}$ where $Q_0$ is a single qubit.
%\vspace{-0.8em}
\paragraph{Bra-ket notation.} We will think of pure states as column vectors.  The pure state that assigns weight 1 to $x$ and weight 0 to each $y\neq x$ is denoted $|x\rangle$.  The set $\{|x\rangle\}$ therefore gives an orthonormal basis for the Hilbert space of pure states.  We will call this basis the ``computational basis.''  If a state $|\phi\rangle$ is a linear combination of several $|x\rangle$, we say that $|\phi\rangle$ is in ``superposition.''  For a pure state $|\phi\rangle$, we will denote the conjugate transpose as the row vector $\langle\phi |$. 
%\vspace{-0.8em}
\paragraph{Entanglement.} In general, a pure state $|\phi\rangle$ over $Q_0\times Q_1$ cannot be expressed as a product state $|\phi_0\rangle\otimes|\phi_1\rangle$ where $|\phi_b\rangle\in Q_b$.  If $|\phi\rangle$ is not a product state, we say that the systems $Q_0,Q_1$ are {\bf entangled}.  If $|\phi\rangle$ is a product state, we say the systems are {\bf un-entangled}.
%\vspace{-0.8em}
\paragraph{Evolution of quantum systems.} A pure state $|\phi\rangle$ can be manipulated by performing a unitary transformation $U$ to the state $|\phi\rangle$.  We will denote the resulting state as $|\phi'\rangle=U|\phi\rangle$.
%\vspace{-0.8em}
\paragraph{Basic Measurements.} A pure state $|\phi\rangle$ can be measured; the measurement outputs the value $x$ with probability $|\langle x|\phi\rangle|^2$.  The normalization of $|\phi\rangle$ ensures that the distribution over $x$ is indeed a probability distribution.  After measurement, the state ``collapses'' to the state $|x\rangle$.  Notice that subsequent measurements will always output $x$, and the state will always stay $|x\rangle$.

If $Q=Q_0\times Q_1$, we can perform a {\bf partial measurement} in the system $Q_0$ or $Q_1$.  If $|\phi\rangle=\sum_{x\in B_0,y\in B_1}\alpha_{x,y}|x,y\rangle$, partially measuring in $Q_0$ will give $x$ with probability $p_x=\sum_{y\in B_1}|\alpha_{x,y}|^2$.  $|\phi\rangle$ will then collapse to the state $\sum_{y\in B_1}\frac{\alpha_{x,y}}{\sqrt{p_x}}|x,y\rangle$.  In other words, the new state has support only on pairs of the form $(x,y)$ where $x$ was the output of the measurement, and the weight on each pair is proportional to the original weight in $|\phi\rangle$.  Notice that subsequent partial measurements over $Q_0$ will always output $x$, and will leave the state unchanged.

The above corresponds to measurement in the computational basis.  Measurements in other bases are possible to, and defined analogously.  We will generally only consider measurements in the computational basis; measurements in other bases can be implemented by composing unitary operations with measurements in the computational basis.
%\vspace{-0.8em}
\paragraph{Efficient Computation.} A quantum computer will be able to perform a fixed, finite set $G$ of unitary transformations, which we will call {\bf gates}.  For concreteness, we will use so-called Hadamard, phase, CNOT and $\pi/8$ gates, but the precise choice is not important for this work, so long as the gate set is ``universal'' for quantum computing.

Let $Q$ be a quantum system on $n$ qubits.  Each gate costs unit time to apply, and each partial measurement also costs unit time.  Therefore, an efficient quantum algorithm will be able to make a polynomial-length sequence of operations, where each operation is either a gate from $G$ or a partial measurement in the computational basis.  Here, ``polynomial'' will generally mean polynomial in $n$.
%\vspace{-2em}
\paragraph{Examples of Quantum Computations.}
\begin{itemize}\setlength\itemsep{-0em}
	\item {\bf Quantum Fourier Transform.} Let $Q_0$ be a quantum system over $B=\Z_q$ for some integer $q$.  Let $Q=Q_0^{\otimes n}$.  The Quantum Fourier Transform (QFT) performs the following operation efficiently: 
	\[\mathsf{QFT}|x\rangle = \frac{1}{\sqrt{q^n}}\omega_q^{x\cdot y}\sum_{y\in\{0,1\}^n}|y\rangle\]

	where $\omega_q=e^{2\pi i/q}$.
	
	\item {\bf Efficient Classical Computations.} Any function that can be computed efficiently classically can be computed efficiently on a quantum computer.  More specifically, if $f$ is computable by a polynomial-sized circuit, then there is a efficiently computable unitary $U_f$ on the quantum system $Q=Q_{in}\otimes Q_{out}\otimes Q_{work}$ with the property that: $U_f |x,y,0\rangle = |x,y+f(x),0\rangle$.
	
	Here, $Q_{in}$ is a quantum system over the set of possible inputs, $Q_{out}$ is a quantum system over the set of possible outputs, and $Q_{work}$ is another quantum system that is just used for workspace, and is reset after use.

\end{itemize}
\vspace{-1.6\topsep}
\paragraph{Quantum Queries.} If a quantum algorithm makes queries to some function $f$, there are two scenarios we will consider.  In one, oracle accepts a quantum state consisting of input and response registers, creates it's own workspace registers initialized to $|0\rangle$, and applies the unitary $U_f$ as defined above to the joint state.  After the query, it returns the input and response registers and discards its workspace registers.  Following~\cite{ICITS:DFNS13}, we call this a \emph{supplied response} oracle.  In the other scenario, the oracle accepts a quantum state consisting of just input registers, creates it's own response and workspace registers initialized to $|0\rangle$, applies the unitary $U_f$, and returns the input and new response registers.  It discards its workspace registers.  We call this a \emph{created response} oracle.  Unless otherwise stated, we will use the supplied response version of every oracle.
\vspace{-0.6em}
\paragraph{Mixed states.} A quantum system may, for example, be in a pure state $|\phi\rangle$ with probability $1/2$, and a different pure state $|\psi\rangle$ with probability $1/2$.  This can occur, for example, if a partial measurement is performed on a product system.%, but then the output of the measurement is erased or forgotten.  

This probability distribution on pure states cannot be described by a pure state alone.  Instead, we say that the system is in a {\bf mixed state}.  The statistical behavior of a mixed state can be captured by {\bf density matrix}.  If the system is in pure state $|\phi_i\rangle$ with probability $p_i$, then the density matrix for the system is defined as $\rho=\sum_i p_i |\phi_i\rangle\langle\phi_i|$.

The density matrix is therefore a positive semi-definite complex  Hermitian matrix with rows and columns indexed by the elements of $B$.  The density matrix for a pure state $|\phi\rangle$ is given by the rank-1 matrix $|\phi\rangle\langle\phi|$.  Any probability distribution over classical states can also be represented as a density matrix, namely the diagonal matrix where the diagonal entries are the probability values.
\vspace{-0.6em}
\paragraph{Distance.}We define the Euclidean distance $\||\phi\rangle-|\psi\rangle\|$ between two states as the value $\left(\sum_x |\alpha_x-\beta_x|^2\right)^{\frac{1}{2}}$ where $|\phi\rangle=\sum_x \alpha_x|x\rangle$ and $|\psi\rangle=\sum_x\beta_x|x\rangle$. 	

We will be using the following lemma:
\begin{lemma}[\cite{BBBV97}]\label{lemma:distance} Let $|\varphi\rangle$ and $|\psi\rangle$ be quantum states with Euclidean distance at most $\epsilon$. Then, performing the same measurement on $|\varphi\rangle$ and $|\psi\rangle$ yields distributions with statistical distance at most
	$4\epsilon$.\end{lemma}

\subsection{Public Key Quantum Money}

Here, we define public key quantum money.  We will slightly modify the usual definition~\cite{CCC:Aaronson09}, though the definition will be equivalent to the usual definition under simple transformations.
\begin{itemize}
	\item We only will consider what Aaronson and Christiano~\cite{STOC:AarChr12} call a quantum money \emph{mini-scheme}, where there is just a single valid banknote.  It is straightforward to extend to general quantum money using a signature scheme
	\item We will change the syntax to more closely resemble our eventual quantum lightning definition, in order to clearly compare the two objects.
\end{itemize}

\noindent A quantum money scheme consists of two quantum polynomial time algorithms $\genmoney,\vermoney$.
\begin{itemize}
	\item $\genmoney$ takes as input the security parameter, and samples a quantum banknote $|\$\rangle$
	\item $\vermoney$ verifies a banknote, and if the verification is successful, produces a serial number for the note.
\end{itemize}

\smallskip

For correctness, we require that verification always accepts money produced by $\genmoney$.  We also require that verification does not perturb the money.  Finally, since $\vermoney$ is a quantum algorithm, we must ensure that multiple runs of $\vermoney$ on the same money will always produce the same serial number.  This is captured by the following two of requirements:
\begin{itemize}
	\item For a money state $|\$\rangle$, let \[H_\infty(|\$\rangle)=-\log_2 \min_s \Pr[\vermoney(|\$\rangle)=s]\] be the min-entropy of $\serialnumber$ produced by applying $\vermoney$ to $|\$\rangle$, were we do not count the rejecting output $\bot$ as contributing to the min-entropy.  We insist that $\E[H_\infty(|\$\rangle)]$ is negligible, where the expectation is over $|\$\rangle\gets\genmoney(1^\lambda)$.  This ensures the serial number is essentially a deterministic function of the money. 
	\item For a money state $|\$\rangle$, let $|\psi\rangle$ be the state left over after running $\vermoney(|\$\rangle)$.  We insist that $\E[|\langle\psi|\$\rangle|^2]\geq 1-\negl(\lambda)$, where the expectation is over $|\$\rangle\gets\genmoney(1^\lambda)$, and any affect $\vermoney$ has on $|\psi\rangle$.  This ensures that verification does not perturb the money.
\end{itemize}

\begin{remark} We note that it is sufficient to only consider the first requirement.  Since the serial number is essentially a deterministic function of the money, we can always modify a $\vermoney$ that does not satisfy the second requirement into an algorithm $\vermoney'$ that does.  $\vermoney'$ runs $\vermoney$, and copies the output $\serialnumber$ into a separate register.  Since $\serialnumber$ is almost deterministic, the copying into a separate register only negligibly affects the money.  Therefore, we un-compute $\vermoney$, and the result will be negligibly close to the original state.
\end{remark}

\noindent For security, consider the following game between an adversary $A$ and a challenger
\begin{itemize}
	\item The challenger runs $\genmoney(1^\lambda)$ to get a banknote $|\$\rangle$.  It runs $\vermoney$ on the banknote to extract a serial number $\serialnumber$.
	\item The challenger sends $|\$\rangle$ to $A$.
	\item $A$ produces two candidate quantum money states $|\$_0\rangle,|\$_1\rangle$, which are potentially entangled. 
	\item The challenger runs $\vermoney$ on both states, to get two serial numbers $s_0,s_1$.
	\item The challenger accepts if and only if both runs of $\vermoney$ pass, and the serial numbers satisfy $s_0=s_1=s$.
\end{itemize}

\begin{definition} A quantum money scheme $(\genmoney,\vermoney)$ is secure if, for all quantum polynomial time adversaries $A$, the probability the challenger accepts in the above experiment is negligible.
\end{definition}

%% file: qlightning.tex
\label{sec:qlightning}

\subsection{Definitions}

The central object in a quantum lightning system is a lightning bolt, a quantum state that we will denote as $\qlightning$.  Bolts are produced by a storm, $\genbolt$, a polynomial time quantum algorithm which takes as input a security parameter $\lambda$ and samples new bolts.  Additionally, there is a quantum  polynomial-time bolt verification procedure, $\verbolt$, which serves two purposes.  First, it verifies that a supposed bolt is actually a valid bolt; if not it rejects and outputs $\bot$.  Second, if the bolt is valid, it extracts a fingerprint/serial number of the bolt, denoted $\serialnumber$.  

Rather than having a single storm $\genbolt$ and single verifier $\verbolt$, we will actually have a family $\Fs_\lambda$ of $(\genbolt,\verbolt)$ pairs for each security parameter.  We will have a setup procedure $\qsetup(1^\lambda)$ which samples a $(\genbolt,\verbolt)$ pair from some distribution over $\Fs_\lambda$.

%We will generally think of $\genbolt,\verbolt$ as not being fixed, but instead being generated during a setup procedure $\qsetup$.  $\qsetup$ is a polynomial-time algorithm that takes as input the security parameter, and produces an instance of the quantum lightning system, namely classical descriptions of $\genbolt,\verbolt$.  

\smallskip

For correctness, we have essentially the same requirements as quantum money.  We require that verification always accepts bolts produced by $\genbolt$.  We also require that verification does not perturb the bolt.  Finally, since $\verbolt$ is a quantum algorithm, we must ensure that multiple runs of $\verbolt$ on the same bolt will always produce the same fingerprint.  This is captured by the following two of requirements:
\begin{itemize}
	\item For a bolt $\qlightning$, let \[H_\infty(\qlightning,\verbolt)=-\log_2 \min_s \Pr[\verbolt(\qlightning)=s]\] be the min-entropy of $\serialnumber$ produced by applying $\verbolt$ to $\qlightning$, were we do not count the rejecting output $\bot$ as contributing to the min-entropy.  We insist that $\E[H_\infty(\qlightning,\verbolt)]$ is negligible, where the expectation is over $(\genbolt,\verbolt)\gets\qsetup(\lambda)$ and $\qlightning\gets\genbolt$.  This ensures the serial number is essentially a deterministic function of the bolt. 
	\item For a bolt $\qlightning$, let $|\psi\rangle$ be the state left over after running $\verbolt(\qlightning)$.  We insist that $\E[|\langle\psi\qlightning|^2]\geq 1-\negl(\lambda)$, where the expectation is over $(\genbolt,\verbolt)\gets\qsetup(\lambda)$, $\qlightning\gets\genbolt$, and any affect $\verbolt$ has on $|\psi\rangle$.  This ensures that verification does not perturb the bolt.
\end{itemize}

\begin{remark} We note that it is sufficient to only consider the first requirement.  Since the serial number is essentially a deterministic function of the bolt, we can always modify a $\verbolt$ that does not satisfy the second requirement into an algorithm $\verbolt'$ that does.  $\verbolt'$ runs $\verbolt$, and copies the output $\serialnumber$ into a separate register.  Since $\serialnumber$ is almost deterministic, the copying into a separate register only negligibly affects the bolt.  Therefore, we un-compute $\verbolt$, and the result will be negligibly close to the original state.
\end{remark}

\noindent For security, informally, we ask that no adversarial storm $\advlightning$ can produce two bolts with the same serial number.  More precisely, consider the following experiment between a challenger and a malicious bolt generation procedure $\advlightning$:
\begin{itemize}
	\item The challenger runs $(\genbolt,\verbolt)\gets\qsetup(1^\lambda)$, and sends $(\genbolt,\verbolt)$ to $\advlightning$.
	\item $\advlightning$ produces two (potentially entangled) quantum states $\qlightning[0],\qlightning[1]$, which it sends to the challenger
	\item The challenger runs $\verbolt$ on each state, obtaining two fingerprints $\serialnumber_0,\serialnumber_1$.  The challenger accepts if and only if $\serialnumber_0=\serialnumber_1\neq\bot$.
\end{itemize}

\begin{definition} A quantum lightning scheme has \emph{uniqueness} if, for all polynomial time adversarial storms $\advlightning$, the probability the challenger accepts in the game above is negligible in $\lambda$.
\end{definition}

\paragraph{Comparison to Public Key Quantum Money.} We note that our quantum lightning definition is very similar to the quantum money notion, except that the security notion is strengthened, and we allow a family of generation/verification procedures.  The differences are analogous to the various notions of security for hash functions $H$:
\begin{itemize}
	\item Quantum money can be seen as an analog of second-pre-image resistance.  Here, a random input $x$ is sampled, hashed to get $y$, and $(x,y)$ are sent to the adversary.  The adversary has to find a second $x'$ that hashes to $y$.  In quantum money, a random note is created, a serial number is extracted, and the adversary must find a second note with the same serial number.
	\item Quantum lightning can then be seen as the analog of collision-resistance.  Here, the adversary just tries to devise two arbitrary distinct  inputs $x,x'$ that hash to the same value.  In quantum lightning, the adversary tries to construct two bolts with the same serial number.  Just as in the collision resistance setting, there are definitional issues with collision resistance that lead the usual (theoretical) definitions to consist of families of hash functions.  Most theoretical constructions of collision-resistant hash functions are also function families.  For similar reasons, we define quantum lightning as a family of storm/verifier pairs.
	\item One can also consider one-wayness, where the adversary is given a random $y$ (but not the pre-image $x$), and the goal is to find an arbitrary pre-image (potentially $x$ itself).  For quantum lightning/money, this would correspond to giving the adversary a random serial number, and then asking the adversary to find some state belonging to that serial number.  We note that this version is trivial without relying on no-cloning: we define the serial number for a state $|x\rangle$ to simply be the hash of $x$.  Then one-wayness already immediately implies that it is hard to find a state with a given serial number.  Therefore, in the context of this paper, such a notion for quantum money/lightning is uninteresting.
\end{itemize}

\paragraph{Variations.}  We consider several variations of the above notion
\begin{itemize}
	\item {\bf No setup.} Here, the set $\Fs_\lambda$ contains only a single $(\genbolt,\verbolt)$ pair.  This means $\qsetup$ simply needs to output the security parameter, and $\genbolt,\verbolt$ are deterministically derived from the security parameter.  
	\item {\bf Common random string.}  Here, each member of $\Fs_\lambda$ is indexed by a bit string $r$ of length $n(\lambda)$.  In this case, $\qsetup$ simply outputs a random string of length $n(\lambda)$, and $\genbolt,\verbolt$ are deterministically derived from this string.
	
	This is in contrast to general $\Fs_\lambda$, where the generation of $(\genbolt,\verbolt)$ may involve secrets that are subsequently discarded.

	\item {\bf Min-entropy.} Here, we consider a slightly different, but closely related, security notion, which basically says that any for any malicious sampling procedure for bolts, the min-entropy of the serial number must be high.  Consider a malicious bolt generator $\advlightning$.  Define \[H_\infty(\advlightning,\verbolt)=-\log \max_\serialnumber\Pr[\verbolt(\qlightning)=\serialnumber:\qlightning\gets\advlightning]\]
	
	to be the min-entropy of serial numbers among the valid bolts generated by $\advlightning$ (bolts that are rejected by $\verbolt$ do not count toward min-entropy).  Note that this is different that $H_\infty(\qlightning,\verbolt)$, which measures the min-entropy of the serial number for a single bolt.  
	
	We say a quantum lightning scheme has \emph{min-entropy} if, for all efficient quantum $\adv$, which takes as input $\genbolt,\verbolt$ and outputs a classical description of $\advlightning$, and for all polynomials $p$
	
	\[\Pr[H_\infty(\advlightning,\verbolt])\leq\log p(\lambda):\advlightning\gets\adv(\genbolt,\verbolt),(\genbolt,\verbolt)\gets\qsetup(1^\lambda)] < \negl(\lambda)\]
		
	In other words, except with negligible probability, serial numbers produced by $\advlightning$ have super-logarithmic min-entropy.  
	
	\medskip
	
	We can modify the above definition to \emph{$p$-min-entropy}, where we insist on a particular amount of min-entropy: for any efficient quantum $\adv$:
	
	\[\Pr[H_\infty(\advlightning,\verbolt])\leq p(\lambda):\advlightning\gets\adv(\genbolt,\verbolt),(\genbolt,\verbolt)\gets\qsetup(1^\lambda)] < \negl(\lambda)\]
		
	With $p$-min-entropy, we can consider $p$ anywhere from super-logarithmic to slightly less than $n$, the bit-length of serial numbers.  We cannot insist on $n$-min-entropy, since an adversarial storm can run $\genbolt$ several times, until, say, the first bit of the serial number is 0, and then only output this bolt.  The adversarial storm will run $\genbolt$ twice in expectation, is guaranteed to produce a valid bolt, and moreover only outputs serial numbers beginning with a 0.  More generally, an efficient generation procedure can always sample bolts with serial numbers from a distribution of min-entropy $n-O(\log \lambda)$, where $n$ is the length of the serial numbers.  Our requirement will be that this is essentially the only strategy.
	
	We say a quantum lightning scheme has \emph{full min-entropy} if this is essentially the only strategy possible for reducing min-entropy: for all efficient quantum $\adv$, there exists a polynomial $p$ such that
	
	\[\Pr[H_\infty(\advlightning,\verbolt])\leq n(\lambda)-\log p(\lambda):\advlightning\gets\adv(\genbolt,\verbolt),(\genbolt,\verbolt)\gets\qsetup(1^\lambda)] < \negl(\lambda)\]
\end{itemize}

Note that a min-entropy adversary easily gives a uniqueness adversary: simply run the min-entropy storm many times, saving all of the valid bolts that are produced.  Since the min-entropy is logarithmic, after a polynomial number of samples, there will be two with the same serial number; simply output these bolts.  This gives the following theorem:

\begin{theorem} If a quantum lightning scheme has uniqueness, then it also has min-entropy.
\end{theorem}

From now on, we will usually only consider the uniqueness security property.  Therefore, when we say that a quantum lightning scheme is ``secure'', we mean that it has uniqueness.  We will only use the other terms when we need to disambiguate the different security notions.

\subsection{Applications}

\paragraph{Quantum Money.} Quantum lightning easily gives quantum money.  To generate a new banknote, simply run $\genbolt$ and output the obtained bolt $\qlightning$ as the quantum money state.  In the case where $\genbolt$ actually comes from a family, first run $\qsetup$ to get $(\genbolt,\verbolt)$, and then run $\genbolt$ to get $\qlightning$.  The quantum money state is $\qlightning,\genbolt,\verbolt$, and it's serial number is $s,\genbolt,\verbolt$, where $s$ is the serial number of $\qlightning$.

Technically, this just gives a quantum money ``mini-scheme'' where there is a single valid banknote.  This can be converted to a full quantum money scheme using signatures~\cite{STOC:AarChr12}.  

\paragraph{Provable Randomness.} Quantum lightning also gives a way to generate a random string, and prove that it has min-entropy.  Assume that $\qsetup$ has already been run in a trusted way (say, by several organizations running $\qsetup$ using an MPC protocol).  

To generate a new random string, simply run $\genbolt$ to get a bolt $\qlightning$.  The random string will be the serial number of the bolt, obtained using $\verbolt$.  $\qlightning$ will be the ``proof'' that the string has min-entropy.

Suppose that an adversarial storm $\advlightning$ can produce strings and valid proofs where the strings have logarithmic min-entropy.  Then, by running $\advlightning$ a polynomial number of times, eventually one will obtain two identical strings, along with two valid proofs.  This violates the security of the quantum lightning scheme.  Hence, we really do obtain a proof of min-entropy.

Notice that $\qsetup$ was run only once, but can then be used to generate arbitrarily many strings along with proofs of min-entropy.

\medskip

Unfortunately, the min-entropy bound obtained is weak; we can only guarantee a super-logarithmic amount of min-entropy.  If we assume sub-exponential hardness of the quantum lightning scheme (that is, even a sub-exponential-time adversary cannot produce two bolts with the same serial number), then we can get a proof of polynomial min-entropy, though the min-entropy may still be much smaller than the overall length of the random string.

We leave obtaining higher min-entropy as an interesting open problem for future work.  One may hope to use randomness extraction, but analyzing appears difficult.  The reason is that randomness extraction usually assumes a random seed that is independent of distribution being extracted.  For quantum lightning, however, the random seed would be part of the description for $\verbolt$, and therefore known to the adversary.  The adversary can potentially craft its distribution on serial numbers so that the extractor fails with the given seed.

\paragraph{Blockchain-less Cryptocurrency.} Finally, we consider using quantum lightning to obtain blockchain-less cryptocurrency.  A coin is simply a bolt, except that the serial number must hash to a value that begins with a certain number of zeros.  To generate a new coin, simply keep generating bolts until the serial number's hash has the prescribed number of zeros.  

Treating the hash function as a random oracle means that the only way to generate a coin is to actually keep generating bolts until the serial number hashes correctly.  The number of zeros is set so that it takes a moderate amount of time to generate new coins.  This ensures scarcity, a crucial feature of any cryptocurrency.

This cryptocurrency is unlikely to be useful in practice due to a very important limitation.  Namely, as technology gets better, it will be easier and easier to create new coins.  Without any modifications, this will lead to an exponentially increasing supply of coins, and hence rampant inflation.  One option is to keep requiring the hashes to contain more and more zeros, but this will render old coins invalid; with our scheme, it is impossible to distinguish coins made today with coins made last week or last year.  In either case, the result is highly undesirable.  

Notice that current cryptocurrency instantiations avoid these problems, essentially, because it is possible to distinguish new coins from old, due to all coins being recorded on a blockchain.  Hence, it is possible to increase the number of 0's required in a hash to combat inflation.

We leave it as an interesting open problem to fix our protocol.  One hope is to combine quantum lightning with some form of time-released cryptography.  The hope is to actually provide a way to prove that a coin was minted some time in the past, so that it can be accepted using the verification procedure from the time it was minted.

%% file: winwin.tex
\label{sec:winwin}

In this section, we give several win-win results for public key quantum money and quantum lightning.  

Recently, Unruh~\cite{EC:Unruh16} and Garg, Yuen, and Zhandry~\cite{C:GarYueZha17} have shown limitations with prior definitions for commitment schemes, hash functions, and signatures\footnote{Technically, Garg et al. only study message authentication codes, but their discussion applies to signatures as well}.  The problem is that the prior definitions implicitly assume that an adversary capable of producing two objects is able to do so simultaneously.  However, in the quantum setting, it may be possible for an adversary to produce one of two objects, but impossible for it to simultaneously produce both.  

For example, consider the case of signatures, which provide integrity over an insecure channel.  Classically, if an adversary intercepts a signed message and modifies it into a different signed message, it then has two signed messages, the one it received and the one it produced.  Inspired by this, Boneh and Zhandry give the first reasonable definition for security when the adversary sees a superposition of signed messages.  Their notion, roughly, says that such an adversary cannot produce two different signed messages.  Unfortunately, this definition allows for undesirable outcomes.  For example, if the original signed message always begins with the bit 0, it would be desirable for any signed message produced by the adversary to also have the first bit be 0.  However, the Boneh-Zhandry definition allows for an adversary to construct a signed message that begins with 1; since the adversary only ever produces a single signed message, this does not contradict Boneh-Zhandry security.  

To combat such situations, Garg, Yuen, and Zhandry define a much stronger notion of security for signatures that rules out such attacks.  Similar situations arise for commitment schemes and hash functions, and Unruh~\cite{EC:Unruh16} similarly gives definitions that rule out these undesirable settings.  

In this section, we show that these undesirable attacks, if they exist for a particular scheme, actually yield quantum money or quantum lightning.  Thus, any scheme that meets the weaker old security notions either (1) actually also meets the stronger security definitions, or (2) can be used to construct quantum money/lightning, in either case leading to a positive outcome.  Given the difficulty of constructing public key quantum money, we interpret our win-win results to suggest that most natural constructions of primitives actually meet the stronger security properties.

\subsection{Infinity-Often Security} 

Before describing our win-win results, we need to introduce a non-standard notion of security.  Typically, a security statement says that no polynomial-time adversary can win some game, except with negligible probability.  A violation of the security statement is a polynomial-time adversary that can win with \emph{non}-negligible probability; that is, some probability $\epsilon$ that is lower bounded by an inverse-polynomial \emph{infinitely often}.  

Our win-win results are of the form ``either (A) is secure or (B) is secure.''  Unfortunately, one of the two security properties needs to be relaxed slightly.  The reason is that we will use a supposed attack for (A) to yield a verifier for (B) that allows us to prove security.  However, if the attack for (A) only succeeds with non-negligible probability, it's winning probability may frequently be too small to be useful for proving (B).  Instead, we will either treat a break for (A) as yielding an attack with an actual inverse polynomial winning probability (so that it will always be useful), or only guarantee security for (B) infinitely often (basically, in the cases where the attack for (A) was useful).

This motivates the notion of \emph{infinitely often security}.  A scheme has infinitely-often security if, roughly, security holds for an infinite number of security parameters, but not necessarily all security parameters.  More precisely, instead of a poly-time adversary's advantage or success probability being upper-bounded by a negligible function, it is only guaranteed to be bounded infinitely often by a negligible function.  If a scheme is not infinitely-often secure, it means that there is an adversary that has an inverse polynomial advantage (as opposed to non-negligible).  It is straightforward to modify all security notions in this work to infinitely-often variants.

Our win-win results will therefore be phrased as:
\begin{itemize}
	\item ``either (A) is secure or (B) is infinitely-often secure'', and 
	\item ``either (A) is infinitely-often secure or (B) is secure.''
\end{itemize}

\subsection{Collision Resistant Hashing}

A hash function is a function $H$ that maps large inputs to small inputs.  We will considered keyed functions, meaning it takes two inputs: a key $k\in\{0,1\}^\lambda$, and the actual input to be compressed, $x\in\{0,1\}^{m(\lambda)}$.  The output of $H$ is $n(\lambda)$ bits.  For the hash function to be useful, we will require $m(\lambda)\gg n(\lambda)$.

The usual security property for a hash function is collision resistance, meaning it is computationally infeasible to find two inputs that map to the same output.
\begin{definition}$H$ is collision resistant if, for any quantum polynomial time adversary $A$,
	\[\Pr[H(x_0)=H(x_1)\wedge x_0\neq x_1:(x_0,x_1)\gets A(k),k\gets\{0,1\}^\lambda]<\negl(\lambda)\]
\end{definition}

Unruh~\cite{EC:Unruh16} points out weaknesses in the usual collision resistance definition, and instead defines a stronger notion called \emph{collapsing}.  Intuitively, it is easy for an adversary to obtain a superposition of pre-images of some output, by running $H$ on a uniform superposition and then measuring the output.  Collapsing requires, however, that this state is computationally indisitnguishable from a random input $x$.  More precisely, for an adversary $A$, consider the following experiment between $A$ and a challenger
\begin{itemize}
	\item The challenger has an input bit $b$.
	\item The challenger chooses a random key $k$, which it gives to $A$.
	\item $A$ creates a superposition $|\psi\rangle=\sum_x \alpha_x |x\rangle$ of elements in $\{0,1\}^{m(\lambda)}$.
	\item In superposition, the challenger evaluates $H(k,\cdot)$ to get the state $|\psi'\rangle=\sum_x \alpha_x |x,H(k,x)\rangle$
	\item Then, the challenger either:
	\begin{itemize}
		\item If $b=0$, measures the $H(k,x)$ register, to get a string $y$.  The state $|\psi'\rangle$ collapses to $|\psi_y\rangle\propto\sum_{x:H(k,x)=y}\alpha_x|x,y\rangle$
		\item If $b=1$, measures the entire state, to get a string $x,H(k,x)$.  The state $|\psi'\rangle$ collapses to $|x,H(k,x)\rangle$
	\end{itemize}
	\item The challenger returns whatever state remains of $|\psi'\rangle$ (namely $|\psi_y\rangle$ or $|x,H(k,x)\rangle$) to $A$.
	\item $A$ outputs a guess $b'$ for $b$.  Define $\collapseexp_b(A,\lambda)$ as the random variable $b'$.
\end{itemize}

\begin{definition}$H$ is collapsing if, for all quantum polynomial time adversaries $A$, \[|\Pr[\collapseexp_0(A,\lambda)=1]-\Pr[\collapseexp_1(A,\lambda)=1]|<\negl(\lambda)\]
\end{definition}

\begin{theorem}\label{thm:collision} Suppose $H$ is collision resistant.  Then both of the following are true:
	\begin{itemize}
		\item Either $H$ is collapsing, or $H$ can be used to build a quantum lightning scheme that is infinitely often secure.
		\item Either $H$ is infinitely often collapsing, or $H$ can be used to build a quantum lightning scheme that is secure.
	\end{itemize}
\end{theorem}

\begin{proof} Let $A$ be a collapsing adversary; the only difference between the two cases above are whether $A$'s advantage is non-negligible or actually inverse polynomial.  The two cases are nearly identical, but the inverse polynomial case will simplify notation.  We therefore assume that $H$ is not infinitely-often collapsing, and will design a quantum lightning scheme that is secure.
	
Let $A_0$ be the first phase of $A$: it receives a hash key $k$ as input, and produces a superposition of pre-images, as well as it's own internal state.  Let $A_1$ be the second phase of $A$: it receives the internal state from $A_0$, plus the superposition of input/output pairs returned by the challenger.  It outputs 0 or 1.	

Define $q_b(\lambda)=\Pr[\collapseexp_b(A,\lambda)=1]$.  By assumption, we have that $|q_0(\lambda)-q_1(\lambda)|\geq 1/p(\lambda)$ for some polynomial $p$.  We will assume $q_0<q_1$, the other case handled analogously.  

For an integer $r$, consider the function $H^{\otimes r}(k,\cdot)$ which takes as input a string $x\in(\{0,1\}^{m(\lambda)})^r$, and outputs the vector $(H(k,x_1),\dots,H(k,x_r))$.  The collision resistance of $H$ easily implies the collision resistance of $H^{\otimes r}$, for any polynomial $r$.  Moreover, we will use $A$ to derive a collapsing adversary $A^{\otimes r}$ for $H^{\otimes r}$ which has near-perfect distinguishing advantage.  $A^{\otimes r}$ works as follows.

\begin{itemize}
	\item First, it runs $A_0$ in parallel $r$ times to get $r$ independent states $|\psi_i\rangle$, where each $|\psi_i\rangle$ contains a superposition of internal state values, as well as inputs to the hash function.
	\item It assembles the $r$ superpositions of inputs into a superposition of inputs for $H^{\otimes r}$, which it then sends to the challenger.
	\item The challenger responds with a potential superposition over input/output pairs (through the output value in $(\{0,1\}^{n(\lambda)})^r$ is fixed).  
	\item $A^{\otimes r}$ disassembles the input/output pairs into $r$ input/output pairs for $H$.  
	\item It then runs $A_1$ in parallel $r$ times, on each of the corresponding state/input/output superpositions, to get bits $b_1',\dots,b_r'$.
	\item $A^{\otimes r}$ then computes $f=(\sum_i b_i')/r$, the fraction of $b_i'$ that are 1.
	\item If $f>(q_0+q_1)/2$ (in other words, $f$ is closer to $q_1$ than it is to $q_0$), $A$ outputs 1; otherwise it outputs 0.
\end{itemize}

Notice that if $A^{\otimes r}$'s challenger uses $b=0$ (so it only measures the output registers), this corresponds to each $A$ seeing a challenger with $b=0$.  In this case, each $b_i'$ with be 1 with probability $q_0$.  This means that $f$ will be a (normalized) Binomial distribution with expected value $q_0$.  Analogously, if $b=1$, each $b_i'$ will be 1 with probability $q_1$, so $f$ will be a normalized Binomial distribution with expected value $q_1$.  Since $q_1-q_0\geq 1/p(\lambda)$, we can use Hoeffding's inequality to choose $r$ large enough so that in the $b=0$ case, $f<(q_0+q_1)/2=q_0+1/2p(\lambda)$ except with probability $2^{-\lambda}$.  Similarly, in the $b=1$ case, $f>(q_0+q_1)/2=q_1-1/2p(\lambda)$ except with probably $2^{-\lambda}$.  This means $A^{\otimes r}$ outputs the correct answer except with probability $2^{-\lambda}$.

We now describe a first attempt at a quantum lightning scheme:

\begin{itemize}
	\item $\qsetup_0$ simply samples and outputs a random hash key $k$.  This key will determine $\genbolt_0,\verbolt_0$ as defined below.
	\item $\genbolt_0$ runs $A_0^{\otimes r}(k)$, where $r$ is as chosen above and $A_0^{\otimes r}$ represents the first phase of $A^{\otimes r}$.
	
	When $A^{\otimes r}_0$ produces a superposition $|\psi\rangle$ over inputs $x\in\{0,1\}^{rm}$ for $H^{\otimes r}(k,\cdot)$ as well as some private state, $\genbolt_0$ applies $H^{\otimes r}$ in superposition, and measures the result to get $y\in\{0,1\}^{rn}$.
	
	Finally, $\genbolt_0$ outputs the resulting state $\qlightning=|\psi_y\rangle$.
	
	\item $\verbolt_0$ on input a supposed bolt $\qlightning$, first applies $H^{\otimes r}(k,\cdot)$ in superposition to the input registers to obtain $y$, which it measures.  It saves $y$, which will be the serial number for the bolt.
	
	Next, consider two possible tests ${\sf Test}_0$ and ${\sf Test}_1$.  In ${\sf Test}_0$, run $A_1^{\otimes r}$ --- the second phase of $A^{\otimes r}$ --- on the $\qlightning$ and measure the result.  If the result is 1 (meaning $A^{\otimes r}$ guesses that the challenger measured the entire input/output registers), then abort and reject.  Otherwise if the result is 0 (meaning $A^{\otimes r}$ guess that the challenger only measured the output), then it un-computes $A_1^{\otimes r}$.  Note that since we measured the output of $A_1^{\otimes r}$, un-computing does not necessarily return the bolt to its original state. 
	
	${\sf Test}_1$ is similar to ${\sf Test}_0$, except that the input registers $x$ are measured before running $A_1^{\otimes r}$. This measurement is not a true measurement, but is instead performed by copying $x$ into some private registers.  Moreover, the abort condition is flipped: if the result of applying $A_1^{\otimes r}$ is 0, then abort and reject.  Otherwise un-compute $A_1^{\otimes r}$, and similarly ``un-measure'' $x$ by un-computing $x$ from the private registers.
	
	$\verbolt_0$ chooses a random $c$, and applies  ${\sf Test}_c$.  If the test accepts, then it outputs the serial number $y$, indicated that it accepts the bolt.
\end{itemize}

\paragraph{Correctness.}  For a valid bolt, ${\sf Test}_0$ corresponds to the $b=0$ challenger, in which case we know $A_1^{\otimes r}$ outputs 0 with near certainty.  This means $\verbolt$ continues, and when it un-computes, the result will be negligibly close to the original bolt. Similarly, ${\sf Test}_1$ corresponds to the $b=1$ challenger, in which case $A_1^{\otimes r}$ outputs 1 with near certainty.  Un-computing returns the bolt to (negligibly close to) its original state.  For a valid bolt, the serial number is always the same.  Thus, $\genbolt,\verbolt$ satisfy the necessary correctness requirements.  

\paragraph{Security.}  Security is more tricky.  Suppose instead of applying a random ${\sf Test}_c$, $\verbolt_0$ applied both tests.  The intuition is that if $\verbolt$ accepts, it means that the two possible runs of $A_1^{\otimes r}$ would output different results, which in turn means that $A_1^{\otimes r}$ detected whether or not the input registers were measured.  For such detection to even be possible, it must be the case that the input registers are in superposition.  Then suppose an adversarial storm $\advlightning$ generates two bolts $\qlightning[0],\qlightning[1]$ that are potentially entangled such that both pass verification with the same serial number.  Then we can measure both states, and the result will (with reasonable probability) be two distinct pre-images of the same $y$, representing a collision.  By the assumed collision-resistance of $H$ (and hence $H^{\otimes r}$), this will means a contradiction.

The problem with the above informal argument is that we do not know how $A_1^{\otimes r}$ will behave on non-valid bolts that did not come from $A_0^{\otimes r}$.  In particular, maybe it passes verification with some small, but non-negligible success probability.  It could be that after passing ${\sf Test}_0$, the superposition has changed significantly, and maybe is no longer a superposition over pre-images of $y$, but instead a single pre-image.  Nonetheless, if the auxiliary state registers are not those generated by $A_0^{\otimes r}$, it may be that the second test still accepts  --- for example, it may be that if $A^{\otimes r}$'s private state contains a particular string, it will always accept; normally this string would not be present, but the bolt that remains after performing one of ${\sf Test}_c$ may contain this string.  We have to be careful to show that this case cannot happen, or if it does there is still nonetheless a way to extract a collision.  

Toward that end, we only choose a single test at random.  We will first show a weaker form of security, namely that an adversary cannot produce two bolts that are both accepted with probability close to 1 and have the same serial number.  Then we will show how to modify the scheme so that it is impossible to produce bolts that are even accepted with small probability.

Consider a bolt where, after measuring $H(k,\cdot)$, the inputs registers are \emph{not} in superposition at all.  In this case, the measurement in ${\sf Test}_1$ is redundant, and we therefore know that both runs of ${\sf Test}_c$ are the same, except the acceptance conditions are flipped.  Since the choice of test is random, this means that such a bolt can only pass verification with probability at most $1/2$.  

More generally, suppose the bolt was in superposition, but most of the weight  was on a single input $x_0$.  Precisely, suppose that when measuring the $x$ registers, $x_0$ is obtained with probability $1-\alpha$ for some relatively small $\alpha$.   We prove the following:
\begin{claim}\label{claim:measurement} Consider a quantum state $|\phi\rangle$ and a projective partial measurement on some of the registers.  Let $|\phi_x\rangle$ be the state left after performing the measurement and obtaining $x$.  Suppose that some outcome of the measurement $x_0$ occurs with probability $1-\alpha$.  Then $\||\phi_{x_0}\rangle-|\phi\rangle\| < \sqrt{2\alpha}$
\end{claim}
\begin{proof} First, the $|\phi_x\rangle$ are all orthogonal since the measurement was projective.  Let $\Pr[x]$ be the probability that the partial measurement obtains $x$.  It is straightforward to show that $|\phi\rangle=\sum_y \sqrt{\Pr[x]}\beta_x|\phi_x\rangle$ for some $\beta_x$ of unit norm. The overall phase can be taken to be arbitrary, so we can set $\beta_{x_0}=1$.  Then we have $\langle\phi_{x_0}|\phi\rangle = \sqrt{1-\alpha}$.  
		
	This means $\||\phi_{x_0}\rangle -|\phi\rangle\|^2 = 2-2(\langle\phi_{x_0}|\phi\rangle) = 2-2\sqrt{1-\alpha}\leq 2\alpha$ for $\alpha\in[0,1]$.
\end{proof}

Now, suppose for the bolt that ${\sf Test}_0$ passes with probability $t$.  Suppose $\alpha\leq 1/200$.  Then ${\sf Test}_1$ can only pass with probability at most $3/2-t$.  This is because with probability at least $199/200$, the measurement in ${\sf Test}_1$ yields $x_0$.  Applying Claim~\ref{claim:measurement}, the result in this case is at most a distance $\sqrt{2/200}=\frac{1}{10}$ from the original bolt.  In this case, since the acceptance criteria for ${\sf Test}_1$ is the opposite of ${\sf Test}_0$, the probability ${\sf Test}_1$ passes is at most $1-t+\frac{4}{10}$ by Lemma~\ref{lemma:distance}.  Over all then, ${\sf Test}_1$ passes with probability at most $(199/200)\left(1-t+\frac{4}{10}\right)+(1/200)\leq \frac{3}{2}-t$.  

Therefore, since the test is chosen at random, the probability of passing the test is the average of the two cases, which is at most $\frac{3}{4}$ regardless of $t$.  Therefore, for any candidate pair of bolts $\qlightning[0]\qlightning[1]$, either:
\begin{itemize}
	\item[(1)] If the bolts are measured, two different pre-images of the same $y$, and hence a collision for $H^{\otimes r}$, will be obtained with probability at least $1/200$
	\item[(2)] The probability that both bolts accept and have the same serial number is at most $\frac{3}{4}$.  
\end{itemize}

Notice that if $\qlightning[0],\qlightning[1]$ are produced by an adversarial storm $\advlightning$, then event (1) can only happen with negligible probability, else we obtain a collision-finding adversary.  Therefore, we have that for any efficient $\advlightning$, except with negligible probability, the probability that both bolts produced by $\advlightning$ accept and have the same serial number is at most $\frac{3}{4}$.  

In the full scheme, a bolt is simply a tuple of $\lambda$ bolts produced by $\genbolt_0$, and the serial number is the concatenation of the serial numbers from each constituent bolt.  The above analysis show that for any efficient adversarial storm $\advlightning$ that produces two bolt sequences $\qlightning[b]=(\qlightning[b,1],\dots,\qlightning[b,\lambda])$, the probability that both sequences completely accept and agree on the serial numbers is, except with negligible probability, at most $\left(\frac{3}{4}\right)^\lambda$, which is negligible.  Thus we obtain a valid quantum lightning scheme.

\iffalse
This allows us to construct our full scheme.  
\begin{itemize}
	\item $\qsetup(1^\lambda)=\qsetup_0(1^\lambda)$
	\item $\genbolt$ runs $\lambda$ copies of $\genbolt_0$ to get $\lambda$ bolts $\qlightning=(\qlightning[1],\dots,\qlightning[\lambda])$.
	\item $\verbolt$ does the following.  It runs $\verbolt_0$ on each $\qlightning[i]$.  If any run of $\verbolt_0$ rejects, $\verbolt$ aborts and rejects.  Otherwise, it collects the list of serial numbers $(y_i)_i$ as the overall serial number.
\end{itemize}

Now we have that for any candidate pair of bolts produced by an adversary, either:
\begin{itemize}
	\item[(1)] If the bolts are measured, for at least one $i$, with probability at least $1/200$, the result will be two different pre-images of the same $y$, and hence a collision for $H^{\otimes r}$.
	\item[(2)] The probability both bolts are accepted and have the same serial number is at most $\left(\frac{3}{4}\right)^\lambda$, which is exponentially small.
\end{itemize}

Thus any efficient adversary that can with non-negligible probability produce two valid bolts with the same serial number gives an efficient collision-finding adversary.  This completes the proof.\fi\end{proof}

\subsection{One-time Signatures}

A signature scheme consists of three polynomial time classical algorithms $\gen,\sign,\ver$.  $\gen$ is a randomized procedure that takes as input the security parameter and produces a secret key and public key pair $(\pk,\sk)\gets\gen(1^\lambda)$.  $\sign$ takes as input the secret key and a message $m$, and produces a signature $\sigma\gets\sign(\sk,m)$.  Finally, $\ver$ takes as input the public key, a message $m$, and a supposed signature $\sigma$ on $m$, and either accepts or rejects.

A signature scheme is correct if $\ver$ accepts signatures outputted by $\sign$: for \[\Pr[\ver(\pk,m,\sigma)=1:\sigma\gets\sign(\sk,m),(\sk,\pk)\gets\gen(1^\lambda)]\geq 1-\negl(\lambda)\]

For security, we will for simplicity only consider one-time signature schemes where the adversary only receives a single superposition of messages.  Also, following Garg, Yuen, and Zhandry~\cite{C:GarYueZha17}, for this subsection only, we will consider the created response model of a quantum query, where the oracle supplies the response register.  Modeling security in the more common supplied response setting is a more complicated task.  Finally, again for simplicity we assume that the signing function is a deterministic function of the secret key and message; this can be made without loss of generality by using a pseudorandom function to generate the randomness.

\paragraph{Boneh-Zhandry security.}  Boneh and Zhandry~\cite{C:BonZha13} give the following definition of security for signatures in the presence of quantum adversaries.  Let $A$ be a quantum adversary, and consider the following experiment between $A$ and a challenger:
\begin{itemize}
	\item The challenger runs $(\sk,\pk)\gets\gen(\lambda)$, and gives $\pk$ to $A$
	\item $A$ makes a quantum superpositions query to the function $m\mapsto \sign(\sk,m)$
	\item $A$ outputs two classical message/signature pairs $((m_0,\sigma_0),(m_1,\sigma_1))$.  
	\item The challenger accepts and outputs 1 if and only if (1) $m_0\neq m_1$, and (2) $\ver(\pk,m_b,\sigma_b)$ for both $b\in\{0,1\}$.  Denote this output by $\wbzexp(A,\lambda)$.  
\end{itemize}

\begin{definition}[Boneh-Zhandry~\cite{C:BonZha13}] A signature scheme is one-time weakly BZ-secure if, for any quantum polynomial time adversary $A$, $\wbzexp(A,\lambda)$ is negligible.
\end{definition}

We can also consider a stronger variant, where the challenger accepts if $(m_0,\sigma_0)\neq(m_1,\sigma_1)$, ruling out the possibility of producing two signatures on a single message.  Denote the output of this modified experiment by $\sbzexp(A,\lambda)$.

\begin{definition}[Boneh-Zhandry~\cite{C:BonZha13}] A signature scheme is one-time strongly BZ-secure if, for any quantum polynomial time adversary $A$, $\sbzexp(A,\lambda)$ is negligible.
\end{definition}

\paragraph{Garg-Yuen-Zhandry security.} Garg, Yuen, and Zhandry~\cite{C:GarYueZha17} recently give a strengthening of Boneh-Zhandry security, which rules out the types of attacks discussed above that are possible under Boneh-Zhandry security.

Let $A$ be an adversary, and consider the following experiment:
\begin{itemize}
	\item The challenger runs $(\sk,\pk)\gets\gen(\lambda)$, and gives $\pk$ to $A$
	\item $A$ makes a quantum superposition query to the function $m\mapsto \sign(\sk,m)$.  
	\item $A$ outputs a superposition $|\psi\rangle=\sum_{m,\sigma,\aux}\alpha_{m,\sigma,\aux}|m,\sigma,\aux\rangle$ of message/tag pairs as well as auxiliary information.
	\item The challenger runs $\ver$ on the $m,\sigma$ registers using the public key.  This is done in superposition.  The challenger then measures the output of $\ver$.
	\begin{itemize}
		\item If the output is 0, the challenger outputs $\bot$.
		\item Otherwise, the challenger outputs what remains in the registers $m,\sigma,\aux$.  This is the state \[|\psi'\rangle\propto \sum_{\substack{m,\sigma,\aux:}{\ver(\pk,m,\sigma)=1)}}\alpha_{m,\sigma,\aux}|m,\sigma,\aux\rangle\]
		Call this output $\gyzexp(A,\lambda)$.
	\end{itemize}
\end{itemize}

We call $A$ \emph{$m$-respecting} if, after the signing query, $A$ cannot directly modify the $m$ registers, but is allowed to operate on the remaining registers, potentially based on the contents of $m$.  Alternatively, $A$ can replace all registers with a special symbol $\bot$.  This captures the ability of an adversary who intercepts a superposition of signed messages to measure the message, copy the message to some other register, and potentially operate on it's own private space.  Such operations will not affect verification.  Moreover, such an adversary can also throw away the entire superposition, and replace it with arbitrary junk.  However, in this case, verification will reject, so we might as well just have $A$ produce $\bot$.

Analogously, we call $A$ \emph{$(m,\sigma)$-respecting} if we do not allow $A$ to directly modify the $m$ or $\sigma$ registers.  $A$ is still allowed to operate on its remaining registers, potentially based on the contents of $m,\sigma$.  Alternatively, $A$ can replace all registers with $\bot$.

\begin{definition}[Garg-Yuen-Zhandry~\cite{C:GarYueZha17}] A signature scheme is one-time weakly GYZ-secure if, for any quantum polynomial time adversaries $A$, there exists an $m$-respecting quantum polynomial time $S$ such that the following two distributions on states are quantum polynomial-time indistinguishable:
	\[\gyzexp(A,\lambda)\;\;\;\;\text{and}\;\;\;\;\gyzexp(S,\lambda)\]
	
\noindent The scheme is one-time strongly GYZ-secure if $S$ can be taken to be $(m,\sigma)$-respecting.  
\end{definition}

\begin{theorem} Suppose $(\gen,\sign,\ver)$ is one-time weakly (resp. strongly) BZ-secure.   Then both of the following are true:
	\begin{itemize}
		\item The scheme is either one-time weakly (resp. strongly) GYZ-secure, or can be used to build an \emph{infinitely-often} secure  public key quantum money scheme.
		\item The scheme is either \emph{infinitely often} one-time weakly (resp. storngly) GYZ-secure, or can be used to build a secure public key quantum money scheme.
	\end{itemize}
\end{theorem}
\begin{proof} We prove the strong setting, assuming the scheme is not infinitely-often GYZ-secure.  The weak setting is nearly identical, as is the case where the scheme is not (always) GYZ-secure.  Consider an adversary $A$ for one-time strong infinitely-often GYZ security.  Consider the following two algorithms derived from $A$:

\begin{itemize}
	\item Let $B$ be $A$, except that after receiving the result of the signing query, $B$ measures the $(m,\sigma)$ registers before continuing.
	\item Let $S$ be the same as $A$, except for the following two changes.  First, $S$ also measures the $(m,\sigma)$ registers, obtaining $(m_0,\sigma_0)$, which it copies into its own private registers.  Second, 	once $A$ produces its final output, a superposition $|\psi\rangle$ over $m,\sigma,\aux$, $S$ tests if $(m,\sigma)=(m_0,\sigma_0)$.  If so, it outputs whatever state remains in $|\psi\rangle$.  Otherwise, it outputs $\bot$.
\end{itemize}
	
Now, $S$ is a $(m,\sigma)$-respecting adversary.  Since we know $A$ is a one-time strong infinitely-often GYZ adversary, this means there is an efficient distinguisher $D$ such that \[|\Pr[D(\gyzexp(A,\lambda))=1]-\Pr[D(\gyzexp(S,\lambda))=1]|\]
is inverse polynomial.  Define:
\begin{itemize}
	\item $q_0(\lambda)=	\Pr[D(\gyzexp(A,\lambda))=1]$
	\item $q_1(\lambda)=	\Pr[D(\gyzexp(B,\lambda))=1]$
	\item $q_2(\lambda)=	\Pr[D(\gyzexp(S,\lambda))=1]$
\end{itemize}

\begin{claim}$|q_1(\lambda)-q_2(\lambda)|$ is negligible
\end{claim}
\begin{proof} Suppose not, that $|q_1(\lambda)-q_2(\lambda)|\geq \delta(\lambda)$ for a non-negligible $\delta$.  Notice that the only instance in which the challenger outputs anything but $\bot$ is when $(m,\sigma,\aux)\neq\bot$ and if $(m,\sigma)$ is valid.  Moreover, if the output is not $\bot$, then the only different between $B$ and $S$ comes from $|m,\sigma,\aux\rangle$ where $(m,\sigma)\neq (m_0,\sigma_0)$.  Therefore, the final superposition produced by $B$ must have weight at least $\delta$ on $(m,\sigma,\aux)$ where $(m,\sigma)$ is valid and not equal to $(m_0,\sigma_0)$.  Thus, we obtain a one-time strong BZ adversary: simply run $B$, copying the post-signing message/signature pair $(m_0,\sigma_0)$ (which was measured by $B$) into a private register.  Then at the end, measure the state, to obtain $(m_1,\sigma_1)$.  With probability at least $\delta$, $(m_1,\sigma_1)$ is valid an not equal to $(m_0,\sigma_0)$.  Thus output $((m_0,\sigma_0),(m_1,\sigma_1))$ as the forgery.  This adversary has non-negligible probability $\delta$ of succeeding.
\end{proof}

Therefore, we have that $|q_0(\lambda)-q_1(\lambda)|$ is an inverse polynomial quantity $1/p(\lambda)$.  As in Theorem~\ref{thm:collision}, we can repeat the scheme many times in parallel to obtain a new signature scheme and adversary where $q_0(\lambda)\leq 2^{-\lambda}$ and $q_1(\lambda)\geq 1-2^{-\lambda}$.  We will abuse notation, and write $A,B,S$ as the algorithms corresponding to this new obtained adversary with almost perfect distinguishing advantage.  Let $A_0,A_1$ be the two phases of $A$, and similarly $B_0,B_1,S_0,S_1$.  Notice that $A_0=B_0=S_0$.  

We now describe our basic  quantum money scheme $(\genmoney,\vermoney)$, assuming $q_0(\lambda)\leq 2^{-\lambda}$ and $q_1(\lambda)\geq 1-2^{-\lambda}$ as above:
\begin{itemize}
	\item $\genmoney_0(1^\lambda)$ samples $(\sk,\pk)\gets\gen(1^\lambda)$.  It then runs $A_0$ on $\pk$.  When $A_0$ outputs a superposition over internal state values and a signing query, $\genmoney_0$ signs the query with $\sk$, placing the output in newly created registers.  $\genmoney_0$ outputs the entire state of the adversary and result of the signing query, along with the public key, as the banknote $|\$\rangle$.
	\item $\vermoney_0(|\$\rangle,\pk)$ first measures the $|\pk\rangle$ register to obtain $\pk'$.  If $\pk'\neq \pk$, $\vermoney$ rejects.

	Then it runs $\ver$ on the $(m,\sigma)$ registers in superposition, and measures the result.  If $\ver$ rejects, then $\vermoney_0$ rejects.
	
	Next, consider two possible tests.  ${\sf Test}_0$ is the following.  Finish running $A$ by running $A_1$ on on the banknote to get a superposition over $(m,\sigma,\aux)$.  Then simulate the challenger by running $\ver(\pk_i,m,\sigma)$ in superposition; if the result is $0$, abort and reject.  If the result is 1, produce the whatever state $|\psi_i'\rangle$ remains.  Feed the result of the previous step into the distinguisher $D$, to get an output $b$. If $b=1$, then abort and reject.  Otherwise, un-compute all of the preceding steps. 
	
	${\sf Test}_1$ is similar, but modified analogously to the proof of Theorem~\ref{thm:collision}.  Instead of running $A_1$, run $B_1$, for which the only difference is that the $(m,\sigma)$ registers are measured at the beginning.  We also change the acceptance condition: if $b=0$, then abort and reject.  Otherwise, un-compute all of the preceding steps. 
	
	$\vermoney_0$ simply chooses a random $c$, and applies ${\sf Test}_c$.  If the test passes, then $\vermoney_0$ outputs the serial number $\pk$.
\end{itemize}

\paragraph{Correctness.} For a valid banknote, the serial number $\pk$ is clearly a deterministic function of the note.  Moreover, the step of $\vermoney$ where $\ver$ is run will always accept without modifying the quantum money state.  Finally, we claim that either test ${\sf Test}_c$ always accepts and negligibly affects the state.  This is true simply because ${\sf Test}_c$ corresponds to running the challenger on input $b=c$, and the the test accepts exactly if $A,D$ behave as guaranteed.

\paragraph{Security.} Analogous to the proof of Theorem~\ref{thm:collision}, the scheme above is not secure.  However, the same arguments can be made to show that for any for any candidate pair of quantum money states $\qlightning[0]\qlightning[1]$, either:
\begin{itemize}
	\item[(1)] If the banknotes are measured, two different valid message/signature pairs for $\pk$ will be produced with probability at least $1/200$.
	\item[(2)] The probability that both banknotes accept is at most $\frac{3}{4}$.  
\end{itemize}

Notice that case (1) can be used to obtain a BZ-forger.  The adversary gets a public key $\pk$ from the BZ challenger.  Then it constructs a quantum money state with serial number $\pk$; the only step it cannot perform for itself is the signing, which it accomplishes using the BZ signing oracle.  Then it runs the adversary to get two banknotes, which it measures.  Since the signature scheme is assumed to be BZ secure, the probability of (1) occurring must be negligible.  Therefore, for any efficient quantum money adversary, it must be the case that (2) happens, except with negligible probability.

Just as in the proof of Theorem~\ref{thm:collision}, we can shrink the probability both banknotes accept to negligible by running multiple instances of the scheme in parallel.  This completes the proof.\end{proof}

\subsection{Commitment Schemes}

Next, we turn to commitment schemes.  An interactive commitment scheme consists of four interactive classical polynomial-time algorithms $\comm_S,\comm_R,\reveal_R$:
\begin{itemize}
	\item In the commit phase, the sender has a message $m$ and security parameter $\lambda$, and the receiver has no input (except the security parameter).  The sender runs $\comm_S(1^\lambda,m)$ and the receiver runs $\comm_R(1^\lambda)$, which may send multiple messages bank and forth.  At the end of the interaction, $\comm_S$ and $\comm_R$ produce some state $\state_S,\state_R$, respectively, which are the saved state for the next round of communication. 
	\item In the reveal phase, the receiver is given the message $m$.  Then the sender sends an ``opening'' to the message $m$, which is just $\state_S$.  The receiver runs $\reveal_R(m,\state_R,\state_S)$, and either accepts or rejects.
\end{itemize}

We call a commitment scheme \emph{publicly verifiable} if, after the commit phase but before the reveal phase, the receiver publishes $\state_R$; in this case we will still require all of the security properties discussed below to hold even if the sender sees $\state_R$ at this point.  

The classical definition of computational-binding for a commitment scheme is the following, adapted to the quantum setting.  consider an adversary, consisting of an algorithm $\comm_S'$.  Consider the following experiment between this adversary and a challenger:

\begin{itemize}
	\item The adversary runs $\comm_S'$ and challenger runs $\comm_R(1^\lambda)$; the two algorithms interact.  The challenger ensures that every message received is measured before responding, guaranteeing that $\comm_R$ is run classically. 
	\item $\comm_S'$ produces two openings $\state_{S,0}',\state_{S,1}'$ and two messages $m_0,m_1$.  $\comm_R$ produces a state $\state_R$.
	\item The challenger receives $(\state_{S,b}',m_b)$ for $b=\{0,1\}$.  For each $b$, it runs $\reveal_R(m_0,\state_R,\state_{S,b}')$.
	\item The challenger outputs 1 if and only if both runs of $\reveal_R$ for $b=0,1$ accept.
\end{itemize}

\begin{definition} $(\comm_S,\comm_R,\reveal_R)$ is computationally binding if, for all quantum polynomial-time adversaries $\comm_S'$, the probability the challenger accepts in the above game is negligible.
\end{definition}

Recently, Unruh~\cite{EC:Unruh16} offers a stronger definition, called collapse-binding.  Here, consider the following experiment between an binding adversary $\comm_S'$ and a challenger.  Note that Unruh only considers non-interactive schemes, whereas we consider interactive schemes.  Therefore, our definition appears different than his.  However, in the case of non-interactive schemes, our definitions coincide.
\begin{itemize}
	\item The challenger is given an input $c\in\{0,1\}$.
	\item The adversary runs $\comm_S'(1^\lambda)$ and the challenger runs $\comm_R(1^\lambda)$; the two algorithms interact.  The challenger ensures that every message received is measured before responding, guaranteeing that $\comm_R$ is run classically. 
	\item $\comm_S'$ produces a superposition $\sum \alpha_{\state_S',m}|\state_S',m\rangle$ which it sends to the challenger.  It may also produce a private state that is entangled with this superposition.  $\comm_R$ produces a classical state $\state_R$.
	\item The challenger then runs $\reveal_R$ \emph{in superposition} on $\state_R$ and the superposition produced by $\comm_S'$.  It measures the result of the computation.  If $\reveal_R$ rejects, the challenger aborts and rejects.
	
	\item Next, if $c=0$, the challenger does nothing.  If $c=1$, the challenger measures the $m$ register of the adversary's state.
	\item Finally, the challenger outputs everything.  This includes the (potentially collapsed) superposition produced by $\comm_S'$, the adversary's private state (if any), and $\state_R$.  Denote this output by $\collapseexp_c(\comm_S',\lambda)$
\end{itemize}

\begin{definition} A commitment scheme $\comm_S,\comm_R,\reveal_R$ is collapse-binding if, for all polynomial time quantum adversaries $\comm_S'$, $\collapseexp_0(\comm_S',\lambda)$ is computationally indistinguishable from $\collapseexp_1(\comm_S',\lambda)$.
\end{definition}

\begin{theorem}\label{thm:comm} Suppose $(\comm_S,\comm_R,\reveal_R)$ is a \emph{publicly verifiable} computationally binding commitment scheme.  Then both of the following are true:
	\begin{itemize}
		\item The scheme is either collapse-binding, or can be used to build an \emph{infinitely-often} secure  public key quantum money scheme.
		\item The scheme is either \emph{infinitely often} collapse-binding, or can be used to build a secure public key quantum money scheme.
	\end{itemize}
\end{theorem}

\begin{proof}  The proof is analogous to the proofs for hash functions and signatures, and we only sketch the proof here.  
Suppose the commitment scheme is computationally binding, but is not infinitely-often collapse-binding.  Then there is an adversary $\comm_S'$ and a distinguisher $D$ such that \[|\Pr[D(\collapseexp_0(\comm_S',\lambda))=1]-\Pr[D(\collapseexp_1(\comm_S',\lambda))=1]|\]
is greater than an inverse polynomial.  For simplicity, assume that $\Pr[D(\collapseexp_0(\comm_S',\lambda))=1]\leq 2^{-\lambda}$ and $\Pr[D(\collapseexp_1(\comm_S',\lambda))=1]\geq 1-2^{-\lambda}$; the more general case can be handled analogously to the hash function/signature case by repeating many instances in parallel.

To generate a quantum money state, run the experiment $\collapseexp$ with $\comm_S'$ until $\comm_S'$ outputs a superposition $\sum \alpha_{\state_S',m}|\state_S',m\rangle$, plus potentially a private state that is entangled with this superposition.  Output this superposition, the adversary's private state, if any, and $\state_R$ produced by $\comm_R$ (recall that the scheme is publicly verifiable, so binding will hold even if $\comm_R$ is public).

To verify a banknote, choose a random $c\in\{0,1\}$.  Finish running $\collapseexp_c$ by running $\reveal_R$ in superposition on $\state_R$ and the superposition of $(\state_S',m)$ pairs; if $c=0$, this is all that happens, while if $c=1$, measure the message registers afterward.  Then take the output of the experiment, and feed it to $D$.  If the output of $D$ is not $c$, reject.  Otherwise, accept, and output $\state_R$ as the serial number.

Similar to the proofs in the case of signatures and hash functions, if an adversary is able to produce two banknotes that have the same serial number $\state_R$, then one of two things happen:
\begin{itemize}
	\item[(1)] The superposition of messages in one of the banknotes has a noticeable weight on at least two messages.  In this case, measuring the message registers will give different answers with noticeable probability.
	\item[(2)] The banknotes will fail verification with noticeable probability.
\end{itemize}

In the (1) case, we can use the two banknotes to produce two openings $\state_{S,0}',m_0,\state_{S,1}',m_1$ that can be used to reveal to two different messages simultaneously.  This gives a violation of computational binding.  Therefore, under the assumption that the scheme is computationally binding, it must be that (2) happens with overwhelming probability.

This does not give us a full quantum money scheme, but by repeating $\lambda$ times in parallel, the probability a banknote accepts in the (2) case becomes exponentially small, giving a full quantum money scheme.\end{proof}

\subsection{Non-interactive Commitments}

A commitment scheme is non-interactive if the commit phase consists of a single message from the sender to receiver.  In this setting, we usually allow a setup phase before the commit phase, where a common random string $crs$ is chosen.  In such a scheme, there is no $\comm_R$, and $\state_R$ is just the sender's commit message together with the $crs$.  Notice here that a non-interactive scheme is automatically publicly verifiable.

\begin{theorem} Suppose $(\comm_S,\reveal_R)$ is a computationally binding non-interactive commitment scheme.  Then both of the following are true:
	\begin{itemize}
		\item The scheme is either collapse-binding, or can be used to build an \emph{infinitely-often} secure quantum lightning scheme
		\item The scheme is either \emph{infinitely often} collapse-binding, or can be used to build a secure quantum lightning scheme
	\end{itemize}
\end{theorem}

The proof is essentially the same as Theorem~\ref{thm:comm}, and very similar to Theorem~\ref{thm:collision}.  The main difference is that $\qsetup$ generates the $crs$ for the commitment scheme.  This is then the common random string which is used to select $\genbolt,\verbolt$.

We note that interactive commitments do not give bolts, since an adversarial bolt generator can generate bolts that deviate from how the  honest receiver would act.  This would potentially allow the bolt generator to set up the bolt in such a way that it can open the commitment to multiple values and hence create multiple valid bolts.  However, for a non-interactive commitment, the receiver plays no role in generating the bolt, so an adversarial storm has no chance of cheating in this way.  Two bolts with the same serial number in this case can be used to open the commitment to two values, breaking computational binding.

%% file: constr.tex
\label{sec:constr}

\subsection{Hardness Assumption}

Consider a sequence of upper-triangular matrices $\Am_i\in\{0,1\}^{m\times m}$ where $m$ for $i=1,\dots,n$.  Here, $n<m$.  Let $\As=\{\Am_i\}_i$.  Define the function $f_\As:\{0,1\}^m\rightarrow\{0,1\}^n$ defined as $f_\As(x)=(x^T\cdot\Am_i\cdot x)_i$, where operations are taken mod 2.  Since $x^2=x\mod 2$, this captures general degree 2 functions over $\F_2$, with the terms coming from the diagonal being the linear terms.

As shown by~\cite{DingYang08,ITCS:AHIKV17}, the function $f_\As$ is \emph{not} collision resistant when the matrices $\Am_i$ are random upper triangular, with reasonable probability.  Here, we recreate the proof, and also discuss the multi-collision resistance of the function.

To find a collision for $f_\As$, choose a random $\Delta\in\{0,1\}^m$.  We will find a collision of the form $x,x-\Delta$.  The condition that $x,x-\Delta$ collide means \[x^T\cdot\Am_i x^T = (x-\Delta)\cdot \Am_i\cdot (x-\Delta)\]
for all $i$.  Expanding out the right hand side and rearranging, this gives
\[\Delta^T\cdot (\Am_i+\Am_i^T)\cdot x = \Delta^T\cdot\Am_i\cdot\Delta\]

This forms a system of $n$ linear equations in $m$ unknowns for $x$.  Let $B_\Delta$ be the $n\times m$ matrix whose rows are $\Delta^T\cdot(\Am_i+\Am_i^T)$ for $i\in[n]$.  Then as long as $B_\Delta$ has rank $n$, a solution for $x$ is guaranteed.  For random (upper triangular) $\Am_i$ and random $\Delta$, this matrix will be rank $n$ with constant probability.

This attack can be generalized to find multiple colliding inputs.  We consider two variants:
\begin{itemize}
	\item Notice that $k+1$ points will always lie in a $k$-dimensional affine space.  Suppose our goal is to $k+1$ colliding inputs, subject to the requirement that they do not lie in a $(k-1)$-dimensional affine space.  Choose random $\Delta_1,\dots,\Delta_k$.  We will compute and $x$ such that $x,x-\Delta_1,\dots,x-\Delta_k$ form $k+1$ colliding points.  Each $\Delta_j$ generates a system of $n$ equations for $x$ as described above.  Let $B=B_{\Delta_1,\dots,\Delta_k}$ be the matrix consisting of all the rows of $B_{\Delta_j}$ as $j$ varies.  As long as $B$ is full rank, a solution for $x$ is guaranteed.  Again, $B$ will be full rank with constant probability, provided $m\geq kn$, and with overwhelming probability if $m\geq kn+\omega(\log\lambda)$.
	
	Notice that because the $\Delta_j$ were chosen at random, with high probability the $k+1$ colliding inputs will form a $k$-dimensional affine space.  We call such a multi-collision a \emph{non-affine} multi-collision.

	\item Many more colliding inputs are possible if we specifically search for affine spaces full of colliding inputs.  We will choose some $\Delta_1,\dots,\Delta_r$; the precise computation of these will be described later.  Our colliding inputs will have the form $x+\sum_s \alpha_s \Delta_s$ for arbitrary $\alpha_i\in\{0,1\}$.  Our goal is to compute an $x$ such that this holds.  We want that for each $i$,
	
	\[(x-\sum_s \alpha_s \Delta_s)^T\cdot\Am_i\cdot(x-\sum_s \alpha_s \Delta_s)\]

	is constant, independent of the $\alpha_s$.  We therefore expand out the product, and group by monimials in the $\alpha_s$.  Notice that $\alpha_s^2=\alpha_s$.  Therefore, we have three kinds of monomials:
	
	\begin{itemize}
		\item $\alpha_s \alpha_{s'}$.  The coefficient of such a monomial is $\Delta_s^T\cdot(\Am_i+\Am_i^T)\cdot\Delta_{s'}$.  
		\item $\alpha_s$.  These occur in two ways, either by multiplying $\alpha_s$ by an $x$ term, or by squaring $\alpha_s$, since $\alpha_s^2=1$.  Therefore, we have that $\Delta_s^T\cdot\Am_i\cdot\Delta_s-\Delta_s\cdot (\Am_i+\Am_i^T)\cdot x$.
		\item Constant monomial.  This is just $x^T\cdot\Am_i\cdot x$.
	\end{itemize}
	
	For all points of the form $x-\sum_s \alpha_s \Delta_s$ to collide, we need all the coefficients of the monomials $\alpha_s\alpha_{s'}$ and $\alpha_s$ to be zero.  To make the coefficients of $\alpha_s\alpha_{s'}$ zero, we will choose the $\Delta_s$ as follows.  Choose a random $\Delta_1$.  Then choose a random $\Delta_2$ such that $\Delta_1^T\cdot(\Am_i+\Am_i^T)\cdot\Delta_{2}=0$ and $\Delta_2$ is linearly independent of $\Delta_1$.  The solution to the linear system has dimension at least $m-n$, and a solution that is independent of $\Delta_1$ exists if $m-n>1$ (equivalently, $m\geq n+2$).  
	
	Next, choose $\Delta_3$ such that $\Delta_1^T\cdot(\Am_i+\Am_i^T)\cdot\Delta_{3}=\Delta_2^T\cdot(\Am_i+\Am_i^T)\cdot\Delta_{3}=0$ and $\Delta_3$ is linearly independent of $\Delta_1,\Delta_2$.  Such a $\Delta_3$ can be found as long as $m\geq 2n+3$.  Repeat in this way until $\Delta_r$ has been computed.  $\Delta_r$ can be found as long as $m\geq (r-1)n+ r$.  
	
	Now we force the $\alpha_s$ coefficients to be 0 by solving for $x$.  These equations are linear in $x$.  There are $rn$ equations and $m$ unknowns, so if  $m\geq rn$ then the system is full rank with constant probability (and if $m\geq rn+\omega(\log\lambda)$ then the system if full rank with overwhelming probability) and a solution can be found.
	
	\medskip
	
	We can therefore have $2^r$ colliding inputs provided $m\geq r n + \max(0,r-n)$, many more than in the first attack.  However, these points will have many affine relationships.
\end{itemize}

\paragraph{Our Assumption.}  We therefore make the following hardness assumption.  

We say a hash function $f$ is $(k+1)$-non-affine multi-collision resistant ($(k+1)$-NAMCR) if it is computationally infeasible to find $k+1$ non-affine colliding inputs.

\begin{assumption}\label{assump:degreetwo} Let $k=\poly(n)$ and let $m<(k+1/2)n$.  Choose random upper triangular $\Am_i\in\{0,1\}^{m\times m}$ for $i=1,\dots,n$, and let $\As=(\Am_i)_i$.  Then the function $f_\As$ is $2(k+1)$-NAMCR.  More precisely, for any quantum adversary $A$, 
	\[\Pr[(x_1,\dots,x_{2k+2})\text{ collide in $f_{\As}$ and are non-affine}:(x_1,\dots,x_{2k+2})\gets A(\As)]\]
is negligible, where $\As$ is a sequence of random upper triangular matrices.
\end{assumption}

Note that it could even be possible to make a stronger assumption that $f_\As$ is $(k+2)$-NAMCR.  However, is is harder to compute a $2(k+1)$ non-affine multi-collision, so our assumption is weaker.  Nonetheless, it will be sufficient for our purposes.

\subsection{Quantum Lightning}

We now describe our quantum lightning construction.

\paragraph{Parameters.} Our scheme will be parameterized by integers $n,k,m$, do be chosen later.

\paragraph{Setup.}  To set up the quantum lightning scheme, simply choose $n$ random upper-diagonal matrices $\Am_i\in\{0,1\}^{m\times m}$, and set $\As=(\Am_i)_i$.  Output $\As$ as the public key.

\paragraph{Bolt Generation.}  We generate a bolt in the following steps.

\begin{itemize}
	\item Generate the uniform superposition
	\[|\phi_0\rangle=\frac{1}{2^{kn/2}}\sum_{\Delta_1,\dots,\Delta_k}|\Delta_1,\dots,\Delta_k\rangle\]
	\item Write $\mathbf{\Delta}=(\Delta_1,\dots,\Delta_k)$ In superposition, run the computation above that maps $\mathbf{\Delta}$ to the affine space $S_\mathbf{\Delta}$ such that, for all $x\in S$, $f_\As(x)=f_\As(x+\Delta_j)$ for all $j$.  This will be an affine space of dimension at least $m-nk$.  Assuming $m-nk$ is super-logarithmic in $\lambda$, then with high probability over the choice of $\As$, the dimension will be exactly $m-nk$ for all but a negligible fraction of $\mathbf{\Delta}$.  Then construct a uniform superposition of elements in $S_\mathbf{\Delta}$.  The resulting state is then:
	\[|\phi_1\rangle=\sum_{\mathbf{\Delta}}\sum_{x\in S_\mathbf{\Delta}}\frac{1}{2^{kn/2}\sqrt{|S_\mathbf{\Delta}|}}|\Delta,x\rangle\]
		
	\item Next, in superposition, compute $f_\As(x)$, and measure the result to get a string $y$.  The resulting state is
	
	\[|\phi_y\rangle\propto \sum_{\mathbf{\Delta},x\in S_\mathbf{\Delta}:f_\As(x)=y}\frac{1}{|S_\mathbf{\Delta}|}|x,\mathbf{\Delta}\rangle\]
	
	\item Finally, in superposition, compute the maps $(x,\Delta_1,\dots,\Delta_k)$ to $(x,x-\Delta_1,\dots,x-\Delta_k)$.  The resulting state is
	
	\[\qlightning[y]\propto\sum_{\mathbf{\Delta},x\in S_\mathbf{\Delta}:f_\As(x)=y}\frac{1}{|S_\mathbf{\Delta}|}|x,x-\Delta_1,\dots,x-\Delta_k\rangle\]
	
	Output this state as the bolt.
\end{itemize}

\begin{remark}\label{rem:1}We note that the support of this state is \emph{all} vectors $(x_0,\dots,x_k)$ such that $f_\As(x_i)=y$ for all $i\in[0,k]$.  Moreover, for all but a negligible fraction, the weight $|S_\mathbf{\Delta}|$ is the same, and so the weights for these components are the same.  Even more, the total weight of the other points is negligible.  Therefore, the bolt $\qlightning[y]$ is negligibly close to the state

\[\sum_{x_0,\dots,x_k:f_\As(x_i)=y\forall i}|x_0,\dots,x_k\rangle=\left(\sum_{x:f_\As(x)=y}|x\rangle\right)^{\otimes(k+1)}\propto\qlightningp[y]^{\otimes (k+1)}\]
 \[\text{where }\qlightningp[y]\propto\sum_{x:f_\As(x)=y}|x\rangle \]
\end{remark}

\paragraph{Verifying a bolt.}  Full verification of a bolt will run a mini verification on each of the $k+1$ sets of $m$ registers.  Each mini verification will output an element in $\{0,1\}^n\cup\{\bot\}$.  Full verification will accept and output $y$ if and only if each mini verification accepts and outputs the same string $y$.  We now describe the mini verification.

Given Remark~\ref{rem:1}, we will assume that the mini verification, when run on a valid bolt, will be given the state $\qlightningp[y]$ for some $y$.  Out goal is to output $y$ in this case, and for any other state, reject.

Mini verification on a state $|\phi\rangle$ will proceed in two steps.  In the first step, we verify that the $|\phi\rangle$ is in the space spanned by $\qlightningp[z]$ as $z$ varies.  We describe the procedure shortly.  

Now notice that the $\qlightningp[z]$ vectors are all orthogonal.  Therefore, in the second step we determine which $\qlightningp[z]$ vector we have by evaluating $f_\As$ in superposition, and measuring the result to obtain $y$.  Notice that for $\qlightningp[y]$, this does not perturb the state.  Then output $y$ as the serial number

We now turn to projecting onto the span of $\{\qlightning[z]\}_z$.  For $r\in\{0,1\}^n$, consider the state \[|\phi_r\rangle=\frac{1}{2^{m/2}}\sum_x (-1)^{r\cdot f_\As(x)}|x\rangle\]

\begin{claim}\label{claim:equiv} The states $\qlightningp[z], z\in\{0,1\}^n$ and $|\phi_r\rangle, r\in\{0,1\}^n$ span the same subspace of states.
\end{claim}
\begin{proof} First, we show that each $|\phi_r\rangle$ lies in the span of the $\qlightningp[z]$.  Indeed,\[|\phi_r\rangle=\sum_z\frac{C_z}{2^{m/2}} (-1)^{r\cdot z}\qlightningp[z]\]
	where $C_z$ is the normalization factor for $\qlightningp[z]$, namely $\sqrt{|\{x:f_\As(x)=z\}|}$.  This is true since the superposition places equal weight on each $x$, and moreover, the phase for each $x$ is exactly $r\cdot f_\As(x)$.
	
Next, we prove that each $\qlightningp[z]$ lies in the span of the $|\phi_r\rangle$.  First, we notice that \[\qlightning[0]\propto\sum_r |\phi_r\rangle\]
since for $x$ where $f_\As(x)\neq 0$, we have that an equal number of $r$ place weight $1$ and $(-1)$ on $|x\rangle$, meaning the overall coefficient for $|x\rangle$ is 0.  Meanwhile, all $x$ where $f_\As(x)=0$ have equal positive weight.

Next, we claim that \[C_z\qlightningp[z]+C_0\qlightning[0]\propto \sum_{r:r\cdot z=0}|\phi_r\rangle\]
The sum on the left is a uniform superposition of $x$ such that $f_\As(x)\in\{0,z\}$.  For elements outside this set, the sum on the right places weight 0, since an equal number of $r$ contribute weight $1$ and $-1$.  Meanwhile, for $x$ in the set, the sum on the right places the same positive weight.  This completes the proof of the claim.
\end{proof}

Given Claim~\ref{claim:equiv}, it suffices to check that the state is in the span on the $|\phi_r\rangle$.  Consider the process of creating $|\phi_r\rangle$ given $r$.  Start with the state $|0\rangle$, and then perform the Hadamard gate qubit-by-qubit to obtain the uniform superposition of all $x\in\{0,1\}^m$.  Next, introduce a phase $(-1)^{r\cdot f_\As(x)}$ to each element $|x\rangle$ to arrive at $|\phi_r\rangle$.  Given a superposition $\sum_r \alpha_r |r\rangle|0\rangle$, this process can be used to construct the state $\sum_r \alpha_r |r\rangle|\phi_r\rangle$.

Now, suppose we can, given $|\phi_r\rangle$ but not $r$, compute $r$.  This means in particular we can \emph{un}compute $r$.  Thus, the state above can be transformed into $\sum_r \alpha_r |\phi_r\rangle$, which represents an arbitrary state in the span of the $|\phi_r\rangle$.  

Therefore, our verification procedure works as follows.  Given a state, it \emph{uncomputes} the computations above.  For a state in the span, the result is a state of the form $\sum_r \alpha_r |r\rangle|0\rangle$.  For a state not in the span, the last register will have come components that are non-zero.  Therefore, verification measures the last register, and accepts if and only if the result is 0.  If accepting, it recomputes the process above.  If the original state was in the span, the new state will be identical to the original state.

\medskip

Therefore, it suffices to compute $r$ from $|\phi_r\rangle$.  Let $R_i(x) = x^T\cdot\Am_i\cdot x$, which is a degree-2 polynomial in the components of $x$.  It is moreover multilinear since $x_i^2=x_i$.  With this notation, we have that \[|\phi_r\rangle=\sum_x (-1)^{\sum_i r_i R_i(x)}|x\rangle\]

Write $x$ as $b,x'$ for $b\in\{0,1\},x'\in\{0,1\}^{m-1}$.  The phase $r\cdot f_\As(x)$ is a degree-2 polynomial in the components of $x$.  Therefore, we can write the phase as \[\left(\sum_i r_i P_i(x')\right)+\left(\sum_i r_i Q_i(x')\right)b\]
where the $P_i(x')$ are multilinear degree-2 polynomials, and the $Q_i$ are \emph{linear} polynomials.  

We therefore apply the Hadamard gate to the first qubit.  The resulting state is
\[\sum_{x'\in\{0,1\}^{m-1}}(-1)^{\sum_i r_i P_i(x')}|\left(\sum_i r_i Q_i(x')\right),x'\rangle\]

If we were to measure the entire state, we would obtain a random $x'$ (from which we could compute $Q_i(x')$), and the sum $\sum_i r_i Q_i(x')$.  This would give us a known linear combination of the $r_i$.  This of course is not enough to compute the entire $r$ vector.  

Instead, we only measure the very first qubit, obtaining a bit $c_1$.  Without any additional measurements, $c_1$ is an unknown linear combination of the $r_i$.  We therefore, in superposition, compute and measure $Q_i(x')$ for each $i$, obtaining $\ell_i$.  Thus we know that $\sum_{i} \ell_i r_i = c_1$.

The state collapses to 
\[\sum_{x'\in\{0,1\}^{m-1}:Q_i(x')=\ell_i\forall i}(-1)^{\sum_i r_i P_i(x')}|x'\rangle\]

This has almost the form of our original state $|\phi_r\rangle$, as it is a superposition over $x$ where the phase is a degree-2 polynomial in $x$.  We would therefore hope to repeat the process above to generate more constrains on $r$.  However, the support of this superposition is not full since $Q_i(x')=\ell_i$, so the above approach will not work.

Instead, we notice that the support is an \emph{affine} subspace $S$ of $\{0,1\}^{m-1}$, since the equations $Q_i(x')=\ell_i$ are linear.  Therefore, any $x$ in $S$ can be written as $x_0+\sum_j a_j v_j$ for scalars $a_j$, and fixed vectors $v_j$, where $j=1,\dots d$ where $d$ is the dimension of the space.  Note that the vectors $v_j$ and $x_0$ can be computed from the $Q_i$ polynomials and the $\ell_i$, which are all known without any additional measurements.  Given that the equations $Q_i(x')=\ell_i$ are chosen at random, if we assume $m\geq n+\omega(\log\lambda)$, then with overwhelming probability the equations are are independent and therefore $d=m-1-n$.  

We therefore perform the map $x\mapsto a$.  Now the superposition is over all strings of length $m-1-n$, and the phase for $a$ can be written as $(-1)^{\sum_i r_i P_i(x_0+\sum_j a_j v_j)}$.  Let $R_i'(a) = P_i(x_0+\sum_j a_j v_j)$, which is a degree-2 polynomial in $a$.  The resulting state has the form
\[\sum_a (-1)^{\sum_i r_i R_i'(a)}|a\rangle\]

This state does have the desired form to keep generating new linear constraints on $r$.  We therefore repeat the above process $u=n+\omega(\log \lambda)$ times.  Each time, we generate a random linear constrain on $r$ (which has dimension $n$), so with overwhelming probability in $n$, the $u$ constraints will be full rank, allowing us to determine $r$.

Each constraint uses up $n+1$ qubits of the state.  Therefore, in order to carry out the above procedure, we need $m\geq u(n+1)$ (Actually, we need $m\geq u(n+1)+\omega(\log\lambda)$ to ensure success in the last repetition).  We therefore set, for example, $n=\lambda,k=2n$ and $m=kn=2n^2$, which allows for $u=\lfloor 2n^2/(n+1)\rfloor=2(n-1)$ for $n\geq 2$.  

\medskip

Putting everything together, mini verification projects onto the span of the $\qlightningp[z]$ (equivalently the span of the $|\phi_r\rangle$), and then measures $f_\As(x)$ in superposition, obtaining $y$.  Supposing the projection accepted, the state is in the span of $\qlightning[z]$, and after measuring $y$, the state at the end of verification must be exactly $\qlightningp[y]$. Full verification then does this for each of the $k+1$ registers, obtaining a list $y_1,\dots,y_{k+1}$.   Full verification, which operates on a state over $k+1$ registers, accepts if and only if none of the $y_i$ are $\bot$, and moreover if they are all equal to the same $y$.  Then this $y$ is the serial number for the overall bolt.  If full verification accepts and outputs serial number $y$, then the state that remains is the state $\qlightning[y]=\qlightningp[y]^{\otimes (k+1)}$.

\paragraph{Security.} We now prove security.  Consider a quantum adversary $A$ that is given $\As$ and tries to construct two (possibly entangled) bolts $\qlightning[0],\qlightning[1]$.  Assume toward contradiction that with non-negligible probability, verification accepts on both bolts, and outputs the same serial number $y$.  

By our above analysis, if acceptance happens, the resulting state is just $\qlightningp[y]^{\otimes 2(k+1)}$.  We therefore measure the state, obtaining $2(k+1)$ random pre-images of $y$.  Since these pre-images live in a space much larger than $2(k+1)$, we have that with overwhelming probability they are non-affine.  Therefore, if acceptance happens, we obtain a $2(k+1)$ non-affine multi-collision.  Since by assumption verification accepts with non-negligible probability, we are therefore able to obtain a $2(k+1)$ non-affine multi-collision with non-negligible probability.  This violates our hardness assumption.

\begin{theorem} If Assumption~\ref{assump:degreetwo} holds, then the scheme above is a secure quantum lightning scheme.
\end{theorem}

\subsection{Collapse-non-binding Hash Functions}

If $k=0$, our construction above simply has a single copy of $\qlightningp[y]$.  This says that the function $f_\As$ is collapse-non-binding, except that the function is not collision resistant.

We can instead view our function as a collapse-non-binding hash function as follows.  Start with the function $f_\As^{\otimes(k+1)}$.  We now restrict the domain to $k+1$ non-affine collisions for $f_\As$.  On this domain, the function's output will always have the form $(y,y,\dots,y)$, so we can just take the output to be $y$.  Call this function $g$.

On this domain, $g$ is almost collision resistant, except that the $k+1$ elements of the multi-collision input can be permuted to obtain a collision for $g$.  We therefore restrict to multi-collisions in sorted order.

Now our input generation will generate the superposition of $k+1$ colliding inputs as before, and then sort.  Of course, soring is non-reversible, so we need to be careful.  Notice that the superposition of colliding inputs is symmetric, in the sense than any permutation on the $k+1$ inputs will result in the same superposition.  We will therefore actually apply the map $\sum_{\sigma} |x_{\sigma(1)},\dots,x_{\sigma(k+1)}\rangle \mapsto |x_1,\dots,x_{k+1}\rangle$, which is a unitary transformation, assuming the $x_1,\dots,x_{k+1}$ are sorted.  This gives us a superposition of inputs to $g$.  To verify that we are still in superpostion, we undo the map to make a symmetric state again, and then apply the bolt verification above.

This gives a collapse-non-binding hash function $g$, albeit on a restricted domain.  We would like a function $h$ that is collapse-non-binding on $\{0,1\}^o$ for some $o$.  Toward that end, we first consider our inputs as the tuples $(\Delta_1,\dots,\Delta_k,x)$.  This is still a restricted domain, so we think of $x=x_0+\sum_i a_i v_i$, where $v_i$ span the space $S_\mathbf{\Delta}$.  The inputs will actually be $(\Delta_1,\dots,\Delta_k,a)$.  Now the domain is unrestricted.  Note that some care is needed, since the dimension of $S_\mathbf{\Delta}$ varies for some $\mathbf{\Delta}$.  This means the domain length varies somewhat.  However, this can be handled with some care; we omit the details.

%% file: qmoney.tex
\label{sec:qmoney}

In this section, we show that, assuming injective one-way functions exist, applying indisitnguishability obfuscation to Aaronson and Christiano's abstract scheme~\cite{STOC:AarChr12} yields a secure quantum money scheme.  

\subsection{Obfuscation}

The following formulation of indistinguishability obfuscation is due to Garg et al.~\cite{FOCS:GGHRSW13}:

\begin{definition}(Indistinguishability Obfuscation) An \emph{indistinguiability obfuscator} $\iO$ for a circuit class $\{\Cs_\lambda\}$ is a PPT uniform algorithm satisfying the following conditions:
	\begin{itemize}
		\item $\iO(\lambda,C)$ preserves the functionality of $C$.  That is, for any $C\in\Cs_\lambda$, if we compute $C'=\iO(\lambda,C)$, then $C'(x)=C(x)$ for all inputs $x$.
		\item For any $\lambda$ and any two circuits $C_0,C_1\in\Cs_\lambda$ with the same functionality, the circuits $\iO(\lambda,C)$ and $\iO(\lambda,C')$ are indistinguishable.  More precisely, for all pairs of PPT adversaries $(\Samp,D)$, if there exists a negligible function $\alpha$ such that
		\[ \Pr[\forall x, C_0(x)=C_1(x):(C_0,C_1,\sigma)\gets \Samp(\lambda)]>1-\alpha(\lambda) \]
		then there exists a negligible function $\beta$ such that
		\[  \big|\Pr[D(\sigma,\iO(\lambda,C_0))=1]-\Pr[D(\sigma,\iO(\lambda,C_1))=1]\big|<\beta(\lambda) \]
	\end{itemize}
\end{definition}

The circuit classes we are interested in are polynomial-size circuits  ---  that is, when $\Cs_\lambda$ is the collection of all circuits of size at most $\lambda$.  We call an obfuscator for this class an \emph{indistinguishability obfuscator for $P/poly$}.  The first candidate construction of such obfuscators is due to Garg et al.~\cite{FOCS:GGHRSW13}.

When clear from context, we will often drop $\lambda$ as an input to $\iO$ and as a subscript for $\Cs$.

\begin{definition} A \emph{subspace hiding obfuscator} (shO) for a field $\F$ and dimensions $d_0,d_1$ is a PPT algorithm $\shO$ such that:
	\begin{itemize}
		\item {\bf Input.} $\shO$ takes as input the description of a linear subspace $S\subseteq\F^n$ of dimension $d\in\{d_0,d_1\}$.  For concreteness, we will assume $S$ is given as a matrix whose rows form a basis for $S$.
		\item {\bf Output.} $\shO$ outputs a circuit $\hat{S}$ that computes membership in $S$.  Precisely, let $S(x)$ be the function that decides membership in $S$.  Then \[\Pr[\hat{S}(x)=S(x)\forall x:\hat{S}\gets\shO(S)]\geq 1-\negl(n)\]
		\item {\bf Security.} For security, consider the following game between an adversary and a challenger, indexed by a bit $b$.
		\begin{itemize}
			\item The adversary submits to the challenger a subspace $S_0$ of dimension $d_0$
			\item The challenger chooses a random subspace $S_1\subseteq \F^n$ of dimension $d_1$ such that $S_0\subseteq S_1$.  It then runs $\hat{S}\gets\shO(S_b)$, and gives $\hat{S}$ to the adversary
			\item The adversary makes a guess $b'$ for $b$.
		\end{itemize}
		The adversary's advantage is the the probability $b'=b$, minus $1/2$.  $\shO$ is secure if, all PPT adversaries have negligible advantage.
	\end{itemize}
\end{definition}

\begin{theorem}\label{thm:sho}If injective one-way functions exist, then any indistinguishability obfuscator, appropriately padded, is also a subspace hiding obfuscator for field $\F$ and dimensions $d_0,d_1$, as long as $|\F|^{n-d_1}$ is exponential.
\end{theorem}

\begin{proof} We first prove the case where $\F=\F_2$, the finite field on two elements, and where $d_1=d_0+1$.  Consider an adversary $\adv$.  Consider the following hybrid experiments:
	\begin{itemize}
		\item $H_0$: in this hybrid, $\adv$ receives $\iO(S_0)$ from the challenger, corresponding to $b=0$.  $S_0$ is appropriately padded before obfuscating so that all the programs received by $\adv$ in the following hybrids have the same length.
		\item $H_1$: in this hybrid, $\adv$ receives an obfuscation of the following function.  Let $\hat{P}$ be an obfuscation under $\iO$ of the simple program $Z$ that always outputs 0 on inputs in $\F^{n-d_0}$.  Let $\Bm$ be a $(n-d_0)\times n$ matrix whose rows are a basis for $S^\bot$, the space orthogonal to $S$.  This basis can be computed by Gaussian elimination.  Then $\hat{S}$ is the obfuscation under $\iO$ of the function 
		\[Q(x)=\begin{cases}1&\text{if }\Bm\cdot x = 0\\1&\text{if }\hat{P}(\Bm\cdot x)=1\\0&\text{Otherwise}\end{cases}\]
		Since $\hat{P}$ always outputs 0, this program still accepts if and only if the input is in $S$.  Therefore, $H_0$ and $H_1$ are indistinguishable by the security of the outer $\iO$ invocation.
		\item $H_2$: this hybrid is the same as $H_1$, except that $\hat{P}$ is the obfuscation under $\iO$ of the function \[P_y(x)=\begin{cases}1&\text{if }\owf(x)=y\\0&\text{Otherwise}\end{cases}\]
		Here, $\owf$ is an injective one-way function, and $y=\owf(x^*)$ for a random $x^*\in\F^{n-d_0}$.  By our assumption that $|\F|^{n-d_1}$ is exponential, and that $d_0=d_1-1$, we have that the bit-length of $x^*$, namely $n-d_0$, is linear in the security parameter.  Therefore, we can invoke the security of $\owf$.  
		
		Notice that the only point on which $Z$ and $P_y$ differ is $x^*$, and finding $x^*$ requires inverting $\owf$.  Therefore, if $\iO$ was a \emph{differing inputs obfuscator}, the obfuscations of $Z$ and $P_y$ would be indistinguishable.  Since $Z$ and $P_y$ differ in only a single input, the results of~\cite{TCC:BoyChuPas14} show that $\iO$ \emph{is} a differing inputs obfuscator for these circuits.  
		
		Therefore, $H_1$ and $H_2$ are computationally indistinguishable.
		
		Notice now, since $\F=\F_2=\{0,1\}$, that $Q(x)$ decides membership in the subspace $S_1$ of vectors $x$ such that $\Bm\cdot x$ is in the span of $x^*$ (which is just $\{0,x^*\}$).  Except with negligible probability, $x^*\neq 0$, and so $S_1$ has dimension $d_0+1=d_1$ and contains $S_0$.  Moreover the set of dimension-$d_1$ spaces containing $S_0$ is in bijection with the set of non-zero $x^*$.
		
		\item $H_3$.  In this hybrid, a random $x^*$ is chosen, $S_1$ is constructed as above, and then obfuscated.  Since $Q(x)$ decides membership in $S_1$, the programs being obfuscated in $H_2$ and $H_3$ are the same, so these two hybrids are indistinguishable by $\iO$.

		\item $H_4$.  Here, we choose $x^*$ at random, except not equal to 0.  Since $x^*$ comes from a set of size $|\F|^{n-d_0}$ which by assumption is exponential, the two distributions are negligibly close.  Now, the set $S_1$ is a random dimension-$d_1$ space containing $S_0$, so $H_4$ corresponds to the case $b=1$.
	\end{itemize}

	Over larger fields, we have to change the proof.  The reason is that $\{0,x^*\}$ is no longer the same as the span of $x^*$.  This means that the function obfuscated in $H_2$ is not a linear subspace, but the union of two parallel affine spaces.  Instead, assume for simplicity that the first digit of $x^*$ is non-zero.  Over large fields, this happens with overwhelming probability; over small fields, we still know that \emph{some} digit is non-zero with overwhelming probability.  The discussion below can be modified easily to work with any other bit.

	For an input $y=(b,y')$ for a bit $y$, let ${\rm Reduce}(y)=y'/b$ if $b\neq 0$ and ${\rm Reduce}(0,y')=\infty$.  

	In the hybrids above, we modify $Q(x)$ to be 

	\[Q(x)=\begin{cases}1&\text{if }\Bm\cdot x = 0\\1&\text{if }\hat{P}({\rm Normalize}(\Bm\cdot x))=1\\0&\text{Otherwise}\end{cases}\]

	In $H_2$, we will choose a random $x'$ such that the first bit is non-zero.  Then we set $x^*={\rm Reduce}(x')$.  Notice that ${\rm Reduce}(x^*)$ is a random string, so we can still invoke the security of $\owf$.  Moreover, now in $H_2$, $Q(x)$ accepts if and only if ${\rm Reduce}(\Bm\cdot x)\in\{0,x^*\}$, which is the same condition as $\Bm\cdot x\in {\rm Span}(x')$.  
	
	If $x'$ was chosen truly randomly, this would correspond to obfuscating a random space of dimension $d_1$ containing $S_0$.  Unfortunately, we did not choose $x'$ uniformly at random, but conditioned on the first digit being non-zero.  To fix this, we choose $i$ from a geometric distribution with probability $1-\frac{1}{|\F|}$, and then choose a random $x'$ such that the first $i-1$ digits are 0, and the $i$th digit is non-zero.   With overwhelming probability, $i$ will be in $[n]$, and the distribution on $x'$ is therefore statistically close to uniform.  We then modify ${\rm Reduce}$ to divide out by the $i$th digit instead of the first.  The same analysis as above applies in the case of more general $i$; now $S_1$ is (statistically close) to a random subspace containing $S_0$.
	
	\medskip
	
	Finally, to handle more general $d_0,d_1$, we perform a sequence of hybrids, first invoking the above on dimensions $(d_0,d_0+1)$, then $(d_0+1,d_0+2)$, etc.  This completes the proof.\end{proof}

\subsection{New No-Conversion and No-Cloning Theorems}

\paragraph{No-Conversion.} Here, we consider the following general task.  Fix a dimension $d$, two sequences $\Ss_1,\Ss_2$ of states $|\psi_i\rangle$ and $|\phi_i\rangle$ of dimension $d$ for $i=[n]$.  Finally, fix a probability distribution $\Ds$ over $[n]$.  

The goal in the $(\Ss_1,\Ss_2,\Ds)$-Conversion problem is to, given a state $|\psi_i\rangle$ for $i$ sampled from $\Ds$, produce the state $|\phi_i\rangle$.  Consider a mechanism $\Ms$ for this task.  We will use the following measure for how well $\Ms$ solves the conversion problem: let $F^2_{\Ss_1,\Ss_2,\Ds}(\Ms)$ denote the expectation of $|\langle \phi_i | \xi \rangle |^2$ where $i\gets \Ds$, and $\xi\gets \Ms(|\psi_i\rangle)$.  In other words, $F_2$ measures the expected fidelity (squared) between the state $\Ms$ produces and the desired output state.

For a set of vectors $\Ss$, let $\Am_{\Ss}$ be the matrix of inner products between the vectors.  For a probability distribution $\Ds$ over $[n]$, let $\Bm_{\Ds}$ be the $n\times n$ matrix where the $(i,j)$ entry is $\sqrt{p_i p_j}$ where $p_i$ is the probability of $i$.  Then let $\Cm_{\Ss_1,\Ss_2,\Ds}=\Am_{\Ss_1}\circ \Am_{\Ss_2}\circ \Bm_\Ds$, where $\circ$ denotes the Hadamard (point-wise) product.  In other words, the $(i,j)$ entry of $\Cm_{\Ss_1,\Ss_2,\Ds}$ is $\sqrt{p_i p_j}\langle \psi_i | \psi_j \rangle \langle \phi_i | \phi_j \rangle$.  For a positive semi-definite Hermitian matrix $\Cm$, let $\lambda_1(\Cm)$ be the spectral radius (that is, maximum eigenvalue) of $\Cm$.

\begin{theorem}[No-Conversion]\label{thm:noconv} For any CPTP operator $\Ms$, we have that $F^2_{\Ss_1,\Ss_2,\Ds}(\Ms)\leq d\times \lambda_1(\Cm_{\Ss_1,\Ss_2,\Ds})$. \end{theorem}

\begin{proof} Pick a basis $\{|x\rangle\}$ for the system, and write 
\[	|\psi_i\rangle=\sum_x a_{i,x}|x\rangle \;\;\;\;\;\;\;\;\;\;\; |\phi_i\rangle=\sum_x b_{i,x}|x\rangle\]

We will think of $\Ms$ as outputting a mixed state $\rho$.  Then for a fixed $i$, we can write the expectation of $|\langle \phi_i | \xi \rangle |^2$ as $\langle \phi_i|\rho|\phi_i\rangle$.  $\Ms$ is linear, so we write $\Ms(|x\rangle\langle x'|)=\sum_{y,y'}\Mm_{(x,y),(x',y')} |y\rangle\langle y'|$ for coefficients $\Mm_{(x,y),(x',y')}$.  By Choi's theorem, the requirement that $\Ms$ is completely positive is equivalent to $\Mm$ being a positive matrix.  Since $\Ms$ is trace preserving, $\Tr(\Ms(|x\rangle\langle x|))=1$, and therefore $\Tr(\Mm)=d$.  

Thus we have that $F^2_{\Ss_1,\Ss_2,\Ds}(\Ms)$ can be written as 
\[F^2_{\Ss_1,\Ss_2,\Ds}(\Ms)=\sum_i p_i \langle \phi_i |\Ms(|\psi_i\rangle\langle\psi_i|)|\phi_i\rangle=\sum_{i,x,x',y,y'}p_i \Mm_{(x,y),(x',y')}a_{i,x}a_{i,x'}^* b_{i,y}^* b_{i,y'}\]

Notice that this expression is linear in $\Mm$.  Consider the space of positive $\Mm$ with trace $d$, which is a convex space containing the set of valid $\Mm$.  This space is the convex hull of the set of $\Mm$ of the form $\Mm_{(x,y),(x',y')}=d v_{x,y'}v^*_{x',y'}$ such that $\sum_{x,y}|v_{x,y}|^2=1$.  Therefore, it suffices to bound the value of $F^2$ for such matrices:

\[\sum_i p_i \langle \phi_i |\Ms(|\psi_i\rangle\langle\psi_i|)|\phi_i\rangle=d\sum_{i,x,x',y,y'}p_i v_{(x,y')}v^*_{(x,y')}a_{i,x}a_{i,x'}^* b_{i,y}^* b_{i,y'}=d\vv^\dagger\cdot\Em^\dagger\Em\cdot\vv\]

Where $\vv$ is the vector of the $v_{(x,y')}$, and $\Em$ is the matrix \[\Em_{i,(x,y')}=\sqrt{p_i}a_{i,x}b_{i,y'} \]

This expression is bounded by $\lambda_1(d\Em^\dagger\cdot\Em)=d \lambda_1(\Em^\dagger\cdot\Em)$.  Notice that the set of eigenvalues for $\Em^\dagger\cdot\Em$ is the same as $\Em\cdot\Em^\dagger=\Cm_{\Ss_1,\Ss_2,\Ds}$, so $F^2$ is bounded by $d\lambda_1(\Cm_{\Ss_1,\Ss_2,\Ds})$ as desired.  \end{proof}

\paragraph{No-Cloning.} Cloning is the special case of conversion where $|\phi_i\rangle=|\psi_i\rangle\otimes|\psi_i\rangle$.  Thus $\Am_{S_2}=\Am_{S_1}\circ\Am_{S_1}$.  This means $\Cm_{\Ss_1,\Ss_2,\Ds}=\Am_{\Ss_1}^{\circ 3}\circ \Bm_\Ds$, where $\Am^{\circ 3}$ means the 3-times Hadamard product of $\Am$.  Assume the probabilities of each state are equal, so that $\Bm_{\Ds}$ is just $1/n$ in every coordinate.  Then $F^2$ is bounded by $\frac{d}{n}\lambda_1(\Am^{\circ 3})$.  More generally, for the process of duplicating a given state $k$ times, $F^2$ will be bounded by $\frac{d}{n}\lambda_1(\Am^{\circ (k+1)})$.  If we assume $n>d$ and that no two states in the set are the same, then $\Am$ will be a matrix with $1$'s on the diagonal, and entries with norm less than 1 off diagonal.  This means, as we increase $k$, $\Am^{\circ (k+1)}$ will asymptotically approach the $n\times n$ identity matrix.  Thus, as $k$ goes to infinity, $F^2$ approaches $\frac{d}{n}$, which is less than 1.  Note if perfect cloning is possible, then it is possible to make $k$ perfect copies, meaning $F_2$ should be 1.  Thus, we re-cast the traditional no-cloning theorem using our theorem.  Moreover, for any set of states, by analyzing $\Am^{\circ 3}$, it is possible to give concrete bounds on the success probability of cloning.

\paragraph{Example.} Consider the following task.  Let $\F$ be a field.  A random subspace $S\subseteq \F^n$ of dimension $n/2$ is chosen, for an even integer $n$.  Let $|\psi_S\rangle=\frac{1}{|\F|^{n/4}}\sum_{x\in S}|x\rangle$ be the uniform superposition over $S$.  The goal is, given $|\psi_S\rangle$ for a random subset $S$, to copy the state.  We would like to upper-bound $F^2$ for this problem.  

Let $N_{a,b}$ count the number of subspaces of dimension $a$ there are inside $\F^b$.  The matrix $\Bm$ is just the matrix that has $1/N_{n/2,n}$ in ever position.  Meanwhile, $\Am_{\Ss_1}$ is the matrix with rows and columns indexed by subspaces $S$, where \[(\Am_{\Ss_1})_{S,T}=\langle\psi_S|\psi_T\rangle=|\F|^{\dim(S\cap T)-n/2}\]

Then the matrix $\Cm$ is:
\[(\Cm)_{S,T}=\frac{1}{N_{n/2,n}} |\F|^{3\dim(S\cap T)-3n/2}\]

We now seek to upper-bound the maximum eigenvalue $\lambda_1$ of this matrix.  It is not hard to see that the maximum eigenvalue corresponds to the unit vector $v = \frac{1}{\sqrt{N_{n/2,n}}}(1\;\;1\;\;\dots\;\;1)$ that places an equal positive weight on each subspace.  Then $\lambda_1(\Cm)$ is just 
\[\lambda_1(\Cm)=\frac{1}{N_{n/2,n}^2}\sum_{S,T}|\F|^{3\dim(S\cap T)-3n/2}=\frac{1}{N_{n/2,n}}\sum_T |\F|^{3\dim(S_0\cap T)-3n/2}\]
Where $S_0$ is any fixed $n/2$-dimensional subspace $S$.  This last inequality follows due to symmetry: the sum over $T$ for any two different subspaces $S_0,S_1$ is the same.  

Now, the number of $T$ such that $\dim(S_0,T)=k$ is upper bounded by $N_{k,n/2}N_{n/2-k,n}$; this is because any such $T$ is the direct sum of a $T_0\subset S$ of dimension $k$, and a $T_1\subset |\F|^n$ of dimension $n/2-k$.  

Thus, we can upper bound $\lambda_1$ as 

\[\lambda_1(\Cm)\leq \sum_{k=0}^{n/2} |\F|^{3k-3n/2} N_{k,n/2}N_{n/2-k,n}/N_{n/2,n}\]

Now we observe that \[N_{a,b}=(|\F|^b-1)(|\F|^b-|\F|)\cdots(|\F|^b-|\F|^{a-1})= |\F|^{ab}\prod_{i=b-a+1}^b(1-|\F|^{-i})\]

Therefore, \begin{align*}N_{k,n/2}N_{n/2-k,n}/N_{n/2,n}&=|\F|^{-kn/2}\frac{\prod_{i=n/2-k+1}^{n/2}(1-|\F|^{-i})\prod_{i=n/2+k+1}^{n}(1-|\F|^{-i})}{\prod_{i=n/2+1}^{n}(1-|\F|^{-i})}\\
&=|\F|^{-kn/2}\frac{\prod_{i=n/2-k+1}^{n/2}(1-|\F|^{-i})}{\prod_{i=n/2+1}^{n/2+k}(1-|\F|^{-i})}\leq |\F|^{-kn/2}
\end{align*}

Therefore, \[\lambda_1(\Cm)\leq \sum_{k=0}^{n/2} |\F|^{3k-3n/2-kn/2}=|\F|^{-3n/2}\sum_{\ell=0}^{n/2}|\F|^{-\ell (n+1)}\leq 2|\F|^{-3n/2}\]

Now, the dimension $d$ that the state $|\psi_S\rangle$ lives in is $|\F|^n$.  Therefore, applying Theorem~\ref{thm:noconv}, we have that $F^2$ is at most $|\F|^n\times \lambda_1(\Cm)\leq 2\times |\F|^{-n/2}$.

\subsection{Quantum Money from Obfuscation}

Here, we recall Aaronson and Christiano's~\cite{STOC:AarChr12} construction, when instantiated with a subspace-hiding obfuscator.  

\paragraph{Generating Banknotes.} Let $\F=\Z_q$ for some prime $q$.  Let $\lambda$ be the security parameter.  To generate a banknote, choose $n$ a random even integer that is sufficiently large; we will choose $n$ later, but it will depend on $q$ and $\lambda$.  Choose a random subspace $S\subseteq \F^n$ of dimension $n/2$.  Let $S^\bot=\{x:x\cdot y=0\forall y\in S\}$ be the dual space to $S$.  

Let $|\$_S\rangle=\frac{1}{|\F|^{n/4}}\sum_{x\in S}|x\rangle$.  Let $P_0=\shO(S)$ and $P_1=\shO(S^\bot)$.  Output $|\$_S\rangle,P_0,P_1$ as the quantum money state.

\paragraph{Verifying banknotes.} Given a banknote state, first measure the program registers, obtaining $P_0,P_1$.  These will be the serial number.  Let $|\$\rangle$ be the remaining registers.  First run $P_0$ in superposition, and measure the output.  If $P_0$ outputs 0, reject.  Otherwise continue.  Notice that if $|\$\rangle$ is the honest banknote state $|\$_S\rangle$ and $P_0$ is the obfuscation of $S$, then $P_0$ will output 1 with certainty.  

Next, perform the quantum Fourier transform (QFT) to $|\$\rangle$.  Notice that if $|\$\rangle=|\$_S\rangle$, now the state is $|\$_{S^\bot}\rangle$.  

Next, apply $P_1$ in superposition and measure the result.  In the case of an honest banknote, the result is 1 with certainty.  

Finally, perform the inverse QFT to return the state.  In the case of an honest banknote, the state goes back to being exactly $|\$_S\rangle$.

\medskip

The above shows that the scheme is correct.  Next, we argue security:

\begin{theorem}\label{thm:qmoney} If $\shO$ is a secure subspace-hiding obfuscator for $d_0=n/2$ and some $d_1$ such that both $|\F|^{n-d_1}$ and $|\F|^{d_1-n/2}$ are exponentially-large, then the construction above is a secure quantum money scheme.
\end{theorem}

\begin{corollary} If injective one-way functions and iO exist, then quantum money exists
\end{corollary}

\noindent We now prove Theorem~\ref{thm:qmoney}

\begin{proof} We prove security through a sequence of hybrids
\begin{itemize}
	\item $H_0$ is the normal security experiment for quantum money.  Suppose the adversary, given a valid banknote, is able to produce two banknotes that pass verification with probability $\epsilon$.
	\item $H_1$: here, we recall that Aaronson and Christiano's scheme is \emph{projective}, so verification is equivalent to projecting onto the valid banknote state.  Verifying two states is equivalent to projecting onto the product of two banknote states.  Therefore, in $H_1$, instead of running verification, the challenger measures in the basis containing $|\$_S\rangle\times|\$_S\rangle$, and accepts if and only if the output is $|\$_S\rangle\times|\$_S\rangle$.  The adversary's success probability is still $\epsilon$.
	\item $H_2$: Here we invoke the security of $\shO$ to move $P_0$ to a higher-dimensional space.  $P_0$ is moved to a random $d_1$ dimensional space containing $S_0$.
	
	We prove that the adversary's success probability in $H_2$ is negligibly close to $\epsilon$.  Suppose not.  Then we construct an adversary $B$ that does the following.  $B$ chooses a random $d_0=n/2$-dimensional space $S_0$.  It queries the challenger on $S_0$, to obtain a program $P_0$.  It then obfuscates $S_0^\bot$ to obtain $P_1$.  $B$ constructs the quantum state $|\$_{S_0}\rangle$, and gives $P_0,P_1,|\$_{S_0}\rangle$ to $A$.  $A$ produces two (potentially entangled) quantum states $|\$_0\rangle|\$_1\rangle$.  $B$ measures in a basis containing $|\$_{S_0}\rangle\otimes|\$_{S_0}\rangle$, and outputs 1 if and only if $|\$_{S_0}\rangle\otimes|\$_{S_0}\rangle$.
	
	If $B$ is given $P_0$ which obfuscates $S_0$, then $A$ outputs 1 with probability $\epsilon$, since it perfectly simulates $A$'s view in $H_1$.  If $P_0$ obfuscates a random space containing $S_0$, then $B$ simulates $H_2$.  By the security of $\shO$, we must have that $B$ outputs 1 with probability at least $\epsilon-\negl$.  Therefore, in $H_2$, $A$ succeeds with probability $\epsilon-\negl$.
	
	\item $H_3$: Here we invoke security of $\shO$ to move $P_1$ to a random $d_1$-dimensional space containing $S_0^\bot$.  By an almost identical analysis to he above, we have that $A$ still succeeds with probability at least $\epsilon-\negl$.
	
	\item $H_4$.  Here, we change how the subspaces are constructed.  First, a random space $T_0$ of dimension $d_1$ is constructed.  Then a random space $T_1$ of dimension $d_1$ is constructed, subject to $T_0^\bot\subseteq T_1$.  These spaces are obfuscated using $\shO$ to get programs $P_0,P_1$.  Next, a random $n/2$-dimensional space $S_0$ is chosen such that $T_1^\bot\subseteq S_0\subseteq T_0$.  $S_0$ is used to construct the state $|\$_{S_0}\rangle$, which is given to $A$ along with $P_0,P_1$.  Then during verification, the space $S_0$ is used again.
	
	The distribution on spaces is identical to that in $H_3$, to $A$ succeeds in $H_4$ with probability $\epsilon-\negl$.  
\end{itemize}

Since on average over $T_0,T_1$, $A$ succeeds with probability $\epsilon-\negl$, there exist fixed $T_0,T_1, T_0^\bot\subseteq T_1$, such that the adversary succeeds for these $T_0,T_1$ with probability at least $\epsilon-\negl$.

We now construct a no-cloning adversary $C$.  $C$ is given a state $|\$_{S_0}\rangle$ for a random $S_0$ such that $T_1^\bot\subseteq S_0\subseteq T_0$, and it tries to clone $|\$_{S_0}\rangle$.  To do so, it constructs obfuscations $P_0,P_1$ of $T_0,T_1$ using $\shO$, and gives them along with $|\$_{S_0}\rangle$ to $A$.  $C$ then outputs whatever $A$ outputs.  $C$'s probability of cloning, $F^2$, is exactly the probability $A$ succeeds in $H_4$, which is $\epsilon-\negl$.

We therefore seek to bound $F^2$ for this instance of the cloning problem.  For simplicity, we will assume $T_0$ is the space of vectors where the last $n-d_1$ components are 0, and $T_1$ is the space where the first $n-d_1$ components are 0.  The other cases are handled analogously.  $T_1^\bot$ is the space where the final $d_1$ components are 0.  Therefore, $S$ is a random space of dimension $n/2$, subject to the last $n-d_1$ components being zero, and the first $n-d_1$ components are free.  We can therefore ignore the first and last $n-d_1$ components (since the final components are always 0, and the initial components will always be uniform superpositions).  Define $n'=2d_1-n$.  Therefore, we can think of choosing a random subset $W$ of dimension $d_1-n/2=n'/2$ out of a space of dimension $2d_1-n=n'$.   

Our states will be indexed by $W$, and we overload notation and write \[|\$_{W}\rangle=\frac{1}{|\F|^{n'/4}}\sum_{x\in W}|x\rangle\]

This is exactly the example worked out above, so the probability that $A$ succeeds is at most $2|\F|^{-n'/2}=2|\F|^{d_1-n/2}$.  By the assumptions of the theorem, this is exponentially small.  Hence $\epsilon-\negl$ is exponentially small, and therefore $\epsilon$ is negligible.\end{proof}

\subsection{A Signature Scheme}

We also note that the construction above gives rise to a particular signature scheme.  On input a message $m$, choose a random subspace $S$ containing $m$.  Construct the programs $P_0,P_1$ obfuscating $S,S^\bot$, using randomness derived from $S$ (say using a PRF).  Then sign the obfuscated programs using an arbitrary signature scheme to obtain $\sigma$.  Output $P_0,P_1,\sigma$ as the signature on $m$.  

If the underlying signature scheme satisfies the Boneh-Zhandry definition of security, then so does the derived signature scheme.  However, consider the following ``attack.''  Create a uniform superposition of messages, ask the signing oracle for a signature, and then measure the signature.  The result is a tuple $P_0,P_1,\sigma$ representing a subspace $S$, along with a (statistically close to) uniform superposition over all messages $m\in S$.  Using the quantum Fourier transform, it is possible to verify this state $S$ using $P_0,P_1$ as in the quantum money scheme.  In particular, it is possible to distinguish this state from the state obtained by measuring the entire message/tag pair.  This violates the Garg-Yuen-Zhandry security notion.  Thus, we give the first example separating the two definitions for signatures.

\subsection{A simplified Proof in the Black Box Setting}

Aaronson and Christiano~\cite{STOC:AarChr12} prove their scheme secure if the subspaces are given as quantum-accessible oracles.  In order to prove security they had to develop a new adversary method, called the inner-product adversary method.

Here, we show that our proof above can be adapted to give a much simpler proof of security for their scheme in the black-box setting.  

First, we going through the hybrids in Theorem~\ref{thm:sho}, we see that the transitions invoking $\iO$ (namely $H_0$ to $H_1$ and $H_2$ to $H_3$), we just use the fact that identical oracles are indistinguishable.  The only difference between $H_1$ and $H_2$ is that the function $\hat{P}$  goes from being all-zeros to accepting a single random input.  Thus, by the lower bound for Grover search~\cite{BBBV97}, we have that $H_1$ and $H_2$ are indistinguishable, except with probability at most $O(q^2/|\F|^{n-d_0})$, where $q$ is the number of quantum queries.  Finally, $H_3$ and $H_4$ are indistinguishable except with probability $1/|\F|^{n-d_0}$.

This shows that for a $d_0$-dimensional subspace $S$, an oracle for $S$ is indisitnguishable from an oracle for $T$ where $T$ is a random $d_0+1$-dimension space containing $S$.  By indistinguishable, we mean that a $q$ query algorithm has at most a probability $O(q^2 |\F|^{d_0-n})$ of distinguishing.  Now, the constant in the Big-Oh is independent of $d_0$.  Therefore, by performing a simple hybrid, we have that $S$ is indistinguishable from a random $T$ of dimension $d_1$ containing $S$, except with probability \[O\left(\sum_{i=d_0}^{d_1-1}q^2 |\F|^{i-n}\right)\leq O(q^2 |\F|^{d_1-1-n})\]

Next, we plug this result into the proof of Theorem~\ref{thm:qmoney}.  Suppose the adversary copies a quantum state with probability $\epsilon$.  The proof of Theorem~\ref{thm:qmoney} shows that $\epsilon$ is at most $O(q^2|\F|^{-(n+1-d_1)}+|\F|^{-(d_1-n/2)}$.  
	
For $|\F|=2$, this gives $\epsilon\leq O(q^2 2^{n-d_1}+2^{d_1-n/2})$.  The quantity on the right can be minimized by choosing the right value for $d_1$, namely $d_1=\frac{3n}{4}+\log_2(q)$.  This gives $\epsilon\leq O(2^{n/4}q)$.  In other words, we have that $q\geq \Omega(\epsilon 2^{-n/4})$.  Compare this to Aaronson and Christiano's result, which says that $q\geq \Omega(\sqrt{\epsilon}2^{-n/4})$.  
 Thus, our bound matches their bound for constant $\epsilon$, though is slightly worse for small $\epsilon$.  However, the advantage of our proof is its simplicity, only requiring the lower bound for Grover search, and a quantitative version of the no-cloning theorem.